\documentclass[reqno,11pt]{article} 
\usepackage[margin=1in]{geometry}
\usepackage{amsmath,amsthm,amssymb,amscd,amstext,amsfonts}
\usepackage{mathtools}
\usepackage{mathrsfs}
\usepackage{thmtools}
\usepackage{latexsym}
\usepackage{verbatim}
\usepackage{framed}
\usepackage{graphicx}
\usepackage{stmaryrd}
\usepackage{enumerate}
\usepackage{fullpage}
\usepackage{bm}
\usepackage{color}
\usepackage{hyperref}
\usepackage{url}
\usepackage{physics}
\usepackage{multicol}
\usepackage{tikz}
\usetikzlibrary{decorations.pathmorphing}   
\usetikzlibrary{arrows}
\usepackage{makeidx}
\makeindex

\usepackage{caption} 
\newtheorem{theorem}{Theorem}

\newtheorem{corollary}{Corollary}
\newtheorem{lemma}{Lemma}

\newtheorem{definition}{Definition}

\theoremstyle{remark}
\newtheorem*{remark}{Remark}

\renewcommand{\norm}[1]{\left\Vert#1\right\Vert}                     
 
\def\supp{{\mathrm{supp}}}                               

\DeclareMathOperator{\Ex}{\mathbb{E}}           

\DeclareMathOperator{\AND}{AND}

\DeclareMathOperator{\XOR}{XOR}
\DeclareMathOperator{\DISJ}{DISJ}

\DeclareMathOperator{\ent}{H}           
\DeclareMathOperator{\IC}{IC}        
\def\ext{{\mathrm{ext}}}           

\newcommand{\cT}{\mathcal T}

\newcommand{\cX}{\mathcal X}
\newcommand{\cY}{\mathcal Y}
\newcommand{\cZ}{\mathcal Z}

\allowdisplaybreaks

\begin{document}

\title{Trading information complexity for error II: the case of a large error and external information complexity}
\author{Yaqiao Li \thanks{McGill University. \texttt{yaqiao.li@mail.mcgill.ca}.}}
\maketitle

\begin{abstract}
Two problems are studied in this paper. (1) How much external or internal information cost is required to compute a Boolean-valued function with an error at most $1/2-\epsilon$ for a small $\epsilon$? It is shown that information cost of order $\epsilon^2$ is necessary and of order $\epsilon$ is sufficient. (2) How much external information cost can be saved to compute a function with  a small error $\epsilon>0$ comparing to the case when no error is allowed?  It is shown that information cost of order at least $\epsilon$ and at most $h(\sqrt{\epsilon})$ can be saved. Except the $O(h(\sqrt{\epsilon}))$ upper bound, the other three bounds are tight. For distribution $\mu$ that is equally distributed on $(0,0)$ and $(1,1)$, it is shown that $\IC^\ext_\mu(\XOR, \epsilon)=1-2\epsilon$ where $\XOR$ is the two-bit xor function. This equality seems to be the first example of exact information complexity when an error is allowed. 
\end{abstract}

\section{Introduction}

The past two decades has witnessed a successful development of information complexity, often as a tool to lower bound communication complexity defined by Yao \cite{Yao1979}, which is then applied to solve many other algorithm and complexity problems. Earlier application of information theoretical ideas to  communication complexity include \cite{ablayev1996lower,SS}.  One important property of  information complexity is that it subsumes many known lower bounds for communication complexity \cite{ICstrong}. Another, perhaps the key property of  information complexity for its success in a variety of applications is that it enjoys a direct sum theorem: the (internal) information complexity of computing $m$ copies of a problem is $m$ times the information complexity of computing one copy. This property has led to notable applications such as direct sum type theorems for communication complexity  \cite{YaoIC,compress}, polynomial space lower bounds for approximating the $k$th frequency moment in the data stream model  \cite{ds},  better lower bound on randomized decision tree complexity \cite{jayram2003two},  matching exponential lower bound for extension complexity of approximating CLIQUE \cite{twobounds}, and exact communication complexity of the set disjointness function \cite{exactComm}, etc. For such applications, it is critical to understand the information complexity of related functions, even for functions with very small inputs such as $\AND: \{0,1\} \times \{0,1\} \to \{0,1\}$ defined by $\AND(x,y)=1$ if and only if $x=y=1$, and $\XOR: \{0,1\} \times \{0,1\} \to \{0,1\}$ defined by $\XOR(x,y)=1$ if and only if $x \neq y$. In fact, previously mentioned applications \cite{ds, twobounds, exactComm} all involve understanding of various aspects of information complexity of $\AND$. 

Let $\epsilon > 0$ denote an error parameter in computing a problem.  A general theme in algorithms and complexity is to understand how much  computational resource one can save or must use if $\epsilon$ error is allowed in computing a problem. We say an error is small if $\epsilon$ is close to $0$, and large if $\epsilon$ is close to $1/2$. This paper is to understand the information complexity of computing functions when either a small or a large error is allowed. We study the external information complexity and  the internal information complexity that are defined in \cite{compress}. This study has been initiated in \cite{my1}, the present work is the second part.  See Table \ref{table:results} for a classification of the work in \cite{my1} and the present work. Results for the prior-free information complexity defined in \cite{braverman2015interactive} are also given.

\begin{table}[ht!]
\centering
\caption{A summary of results.} \label{table:results}
\begin{tabular}{|c|c|c|}  
\hline
&a small error &a large error \\
\hline
internal information complexity & \cite{my1} & this work \\
\hline
external information complexity  & this work & this work\\
\hline
\end{tabular}
\end{table}

Consider the external information complexity for example.  Let $f:\cX \times \cY \to \{0,1\}$ be a two-party Boolean-valued function, $\mu$ be a probability distribution on $\cX \times \cY$, and $\IC_\mu^\ext(f,\epsilon)$ denote the external information complexity of computing $f$ with point-wise error $\epsilon$ (see definition in Section \ref{sec:def}). Intuitively, this external information complexity is the least amount of information that any $\epsilon$ error algorithm has to reveal about the input. When $\epsilon >0$ is close to $0$, we study the upper and lower bounds for $\IC_\mu^\ext(f,0) - \IC_\mu^\ext(f,\epsilon)$. The techniques in \cite{my1} can be directly applied to show,
\begin{equation}  \label{eq:sample-result-ext-small-err}
\Omega(\epsilon) \le \IC_\mu^\ext(f,0) - \IC_\mu^\ext(f,\epsilon) \le O(h(\sqrt{\epsilon})) 
\end{equation}
where $h(\alpha) = -\alpha \log_2 \alpha - (1-\alpha) \log_2(1-\alpha)$ is the Shannon entropy. The constants in $\Omega(\cdot)$ and $O(\cdot)$ are explicitly given, see Theorem \ref{thm:ext-ub-small-pt-err} and Theorem \ref{thm:ext-trivial-lb-small-pt-err}.
When $\mu$ is a product distribution, we show an $\Omega(h(\epsilon))$ lower bound with an explicit constant in $\Omega(\cdot)$, see Theorem \ref{thm:ext-prod-lb-small-pt-err}.
The proof of Theorem \ref{thm:ext-prod-lb-small-pt-err} is much simpler than the proof for the $\Omega(h(\epsilon))$ lower bound of internal information complexity for arbitrary distributions in \cite[Theorem 3.2]{my1}.

 When $1/2-\epsilon > 0$ is close to $1/2$ (here we view $1/2-\epsilon$ as the error, \emph{not} $\epsilon$ itself), we study the upper and lower bounds for $\IC_\mu^\ext(f,1/2-\epsilon) - \IC_\mu^\ext(f,1/2) = \IC_\mu^\ext(f,1/2-\epsilon)$, since $\IC_\mu^\ext(f,1/2) =0$ for Boolean-valued functions. Under some conditions we show
\begin{equation}  \label{eq:sample-result-ext-large-err}
\Omega(\epsilon^2) \le \IC_\mu^\ext(f,1/2-\epsilon)  \le O(\epsilon)
\end{equation}
where the constant in $\Omega(\cdot)$ and $O(\cdot)$ are explicitly given, see Theorem \ref{thm:large-err-pt-lb}. To show \eqref{eq:sample-result-ext-large-err}, it is convenient to first establish some lower bounds for information costs (see Section \ref{sec:explicit-lb}). These lower bounds bear some similarities with those used in \cite{ds, Hellinger, jayram2003two} where Hellinger distance is used, whereas we use the $L_1$ distance directly combined with the Pinsker inequality. The $L_1$ distance is easy to work with and give tight bounds in Section \ref{sec:large-err}. The lower bounds also enable us to characterize all different types of distributions of inputs (we call \emph{trivial distributions}) under which various information complexity measures vanish (see Section \ref{sec:explicit-lb}). Some of these characterizations have been obtained in \cite{my1}, here we give a uniform treatment of all characterizations, though the technique is essentially the same as in \cite{my1}. One can then establish \eqref{eq:sample-result-ext-large-err} using the lower bounds and characterization of trivial distributions.

The lower bound in \eqref{eq:sample-result-ext-small-err} and both bounds in \eqref{eq:sample-result-ext-large-err} are all tight with respect to the order of $\epsilon$, via existing examples in \cite{my1} and \cite{twobounds}, respectively. We re-examine some of these examples. Let 
$\mu = 
\begin{pmatrix}
1/2 & 0 \\
0 & 1/2
\end{pmatrix}$ 
be a distribution on $\{0,1\} \times \{0,1\}$, i.e., $\mu(0,0)  = \mu(1,1) = 1/2$. Improving an upper bound in \cite{my1}, we show that for every $0 \le \epsilon \le 1/2$, 
\begin{equation}   \label{eq:sample-result-xor}
\IC_\mu^\ext(\XOR,\epsilon)  = 1 - 2\epsilon.
\end{equation}
To the knowledge of the author, this is the \emph{first} non-trivial \emph{exact} information complexity for an explicit function when $\epsilon>0$ (the \emph{exact} information complexity with no error, i.e., $\epsilon=0$, is known for $\AND$ function, see \cite{exactComm}). Furthermore, \eqref{eq:sample-result-xor} holds for \emph{every} $\epsilon$ in $[0,1/2]$. Combining this with \eqref{eq:sample-result-ext-large-err} shows that when $\epsilon>0$ is small,
\begin{equation}  \label{eq:sample-result-xor-pf}
\IC^\ext(\XOR,1/2-\epsilon) = \Theta(\epsilon)
\end{equation}
where $\IC^\ext(\XOR,1/2-\epsilon)$ denotes the prior-free information complexity of $\XOR$ (see definition in Section \ref{sec:def}). In \cite{twobounds}, a protocol is given that computes $\AND$ with point-wise error $1/2-\epsilon$ and external information cost at most $O(\epsilon^2)$. We use Wolfram Mathematica to explicitly compute the external (and internal) information cost of this protocol. Combining this with \eqref{eq:sample-result-ext-large-err} shows that when $\epsilon>0$ is small,
\begin{equation}  \label{eq:sample-result-and-pf}
\IC^\ext(\AND,1/2-\epsilon) = \Theta(\epsilon^2).
\end{equation}
It is not hard to show that $\IC^\ext(\XOR,0) = 2$, while \cite{exactComm} shows that $\IC^\ext(\AND,0) = \log_2 3 < 2$. This exhibits the difference of $\XOR$ and $\AND$ using information complexity (such difference can \emph{not} be seen in the realm of communication complexity, since correctly computing $\XOR$ and $\AND$ both require $2$ bits to be communicated). The results \eqref{eq:sample-result-xor-pf} and \eqref{eq:sample-result-and-pf} extend such difference into the regime where an error is allowed.

Due to the direct sum property of information complexity as mentioned before, these results on information complexity of functions, even functions with small inputs, could be useful in a variety of contexts. For example, before the paper \cite{twobounds}, one barrier to show matching extension complexity for approximating CLIQUE  is because of \eqref{eq:sample-result-and-pf}. It is the author's wish to discover or see more applications.

The rest of the paper is organized as follows. In Section \ref{sec:def} we define the notion of various information complexity measures and prove some simple inequalities. In Section \ref{sec:explicit-lb} we prove lower bounds for information costs and characterize trivial distributions. In Section \ref{sec:large-err} and \ref{sec:ext-small-err} we prove results in trading information complexity for large and small errors. In Section \ref{sec:eg} we study tight examples: $\XOR$ and $\AND$. Lastly,  some open problems are given in Section \ref{sec:prob}.

\medskip

{\noindent \bf Acknowledgement} The author is grateful to Hamed Hatami for valuable discussions, and to Yuval Filmus for pointing out \cite{twobounds}.

\section{Information complexity and some inequalities} \label{sec:def}
We assume familiarity with common information theoretical notions such as Shannon entropy, mutual information, and Kullback-Leibler divergence, all of which can be found in the standard book \cite{cover2012elements}. We will formally define information complexity. More  discussion on communication complexity, information complexity and its applications can be found, e.g., in \cite{Kushilevitz1997Communication}, \cite{braverman2015interactive} and the survey \cite{icm}.

\subsection{Notation from information theory}  \label{sec:estimates}
Let $\ent(X)$ denote the Shannon entropy of a random variable $X$. Sometimes for clarity we write $\ent_\gamma(X)$ to explicitly indicate that the random variable $X$ is distributed according to $\gamma$. Let $I(X;Y)$  denote the mutual information between random variables $X$ and $Y$, and $I(X;Y|Z)$ denote the  mutual information of $X$ and $Y$ conditioned on another random variable $Z$. 
Given two distributions $\mu,\nu$ that are distributed on the same space, the Kullback-Leibler divergence (divergence for short) from $\nu$ to $\mu$ is defined as $D(\mu  \| \nu) = \Ex_{x \sim \mu} \log \frac{\mu(x)}{\nu(x)}$.  Two facts relating mutual information to divergence are: (1) $I(X;Y) = D(p(X,Y)\| p(X)p(Y))$, where $p(X,Y)$ is the joint distribution of $(X,Y)$, and $p(X), p(Y)$ are the distributions of $X$ and $Y$, respectively;  (2) $I(X; Y) = \Ex_Y D(p(X|Y) || p(X))$ where $p(X|Y)$ is the conditional distribution of $X$ conditioned on $Y$.

For the sake of brevity, we often write $xy$ and $XY$ to denote the input $(x,y)$ and $(X,Y)$, respectively.  Let $\Delta(\cX \times \cY)$ denote the set of all probability distributions on $\cX \times \cY$. In this paper, $\mu$ usually denotes a probability distribution (distribution for short) on $\cX \times \cY$, i.e., $\mu \in \Delta(\cX \times \cY)$. Let $\supp\mu$  denote the support of $\mu$.

For every $\epsilon \in [0,1]$, let $h(\epsilon) = -\epsilon \log\epsilon - (1-\epsilon) \log(1-\epsilon)$ denote the \emph{binary entropy}, where here and throughout the paper $\log(\cdot)$ is in base $2$, and $0 \log 0 = 0$.

\subsection{Information complexity and the rectangle property}  \label{sec:defIC}
The two-party communication model was introduced by Yao~\cite{Yao1979} in 1979. In this model there are two players (with unlimited computational power), often called Alice and Bob, who wish to collaboratively compute a function $f\colon \cX \times \cY \to \cZ$.   Alice receives an input $x \in \cX$ and Bob receives an input $y \in \cY$. Neither of them knows the other player's input, and they wish to communicate, by sending binary bits to each other, in accordance with an agreed-upon protocol $\pi$ to  compute $f(x,y)$.

\begin{definition}[The deterministic communication protocol and its transcript, \cite{Kushilevitz1997Communication}]   \label{def:protocol}
A (deterministic) communication protocol $\pi$ over domain $\cX \times \cY$ and range $\cZ$ is a binary tree where each internal node $v$ is labeled either by a function $a_v: \cX \to \{0,1\}$ or by a function $b_v: \cY \to \{0,1\}$, and each leaf is labeled with an element in $\cZ$.

Given an input $(x,y) \in \cX \times \cY$ to a protocol $\pi$, it naturally defines a path in the binary tree from the root to a leaf: at each node $v$, the path goes to left or right according to $a_v(x) =0$ or $a_v(x) = 1$ (resp. $b_v(y)= 0$ or $b_v(y) = 1$). The label of the leaf of this path is the output of $\pi$ on $(x,y)$, denoted by $\pi(x,y)$. The sequence of bits on this path  is called the transcript of $\pi$  on input $(x,y)$, denoted by $\pi_{xy}$.
\end{definition}
Usually we say that Alice owns the nodes labeled by $a_v$, and Bob owns the nodes labeled by $b_v$. When $\pi(x,y) = f(x,y)$ for every $(x,y) \in \cX \times \cY$, the protocol $\pi$ is said to compute $f$ correctly (or simply, compute $f$). Figure \ref{fig:protocol-two-bits} gives a simple protocol that correctly computes $\XOR: \{0,1\} \times \{0,1\} \to \{0,1\}$.  In this example, $a_v(x)  = x$, $b_u(y) = y$ and $b_w(y) = y$, i.e., Alice first sends her private input bit to Bob, and Bob then sends his private input bit to Alice. 

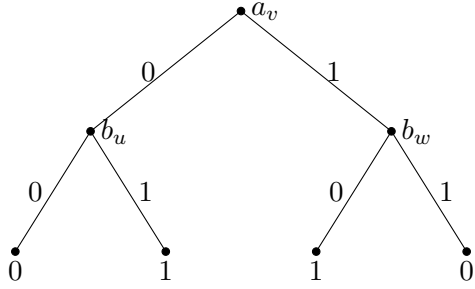
\begin{figure}[h!]   
\begin{center}
\begin{tikzpicture}
\draw [fill] (0,3.2) circle [radius=0.05];
\node [right] at (0,3.2) {$a_v$};

\draw [fill] (2,1.6) circle [radius=0.05];
\node [right] at (2,1.6) {$b_{w}$};

\draw [fill] (-2,1.6) circle [radius=0.05];
\node [right] at (-2,1.6) {$b_{u}$};

\draw [fill] (-1,0) circle [radius=0.05];
\draw [fill] (1,0) circle [radius=0.05];
\draw [fill] (-3,0) circle [radius=0.05];
\draw [fill] (3,0) circle [radius=0.05];

\node [below] at (-3,0) {$0$};
\node [below] at (-1,0) {$1$};
\node [below] at (1,0) {$1$};
\node [below] at (3,0) {$0$};

\node [left] at (-1,2.4) {$0$};
\node [right] at (1,2.4) {$1$};

\node [left] at (-2.5,0.8) {$0$};
\node [right] at (-1.5,0.8) {$1$};

\node [left] at (1.5,0.8) {$0$};
\node [right] at (2.5,0.8) {$1$};

\draw (0,3.2) -- (-2,1.6);
\draw (0,3.2) -- (2,1.6);
\draw (2,1.6) -- (1,0);
\draw (2,1.6) -- (3,0);
\draw (-2,1.6) -- (-1, 0);
\draw (-2,1.6) -- (-3, 0);
\end{tikzpicture}
\end{center}
\caption{A two-bit protocol $\pi$.}
\label{fig:protocol-two-bits}
\end{figure}

Definition  \ref{def:protocol} can be generalized to allow each player to have private access to randomness: Alice has access to a random string $r_A\in R_A$ and Bob has access to a random string $r_B \in R_B$. These two random strings are chosen independently from each other and can have different distributions. 

\begin{definition}[The randomized protocol and its random transcript, \cite{Kushilevitz1997Communication}]   \label{def:protocol-random}
In a randomized protocol, the functions $a_v$ and $b_v$ are random functions. That is,  $a_v: \cX \times R_A \to \{0,1\}$, where $r_A \in R_A$ is Alice's private random string. Similarly for $b_v$.  Let $\pi$ be a randomized  protocol. We use $\Pi(x,y)$ to denote the random output of the protocol $\pi$ on input $(x,y)$, and $\Pi_{xy}$ to denote the random transcript of $\pi$ on input $(x,y)$. 
\end{definition}

Throughout this paper we assume protocols are randomized (a deterministic protocol is a randomized protocols in which the randomness for both Alice and Bob is a point distribution).  The following rectangle property of transcripts is important (see, e.g.,  \cite[Lemma 6.7]{ds}). For completeness we provide a proof. 

\begin{lemma}[Rectangle property, \cite{ds}] \label{lem:rectangle-property}
Let $\pi$ be a randomized protocol,  $v$ be an arbitrary node in the binary tree corresponding to $\pi$. Let $(x_1,y_1)$ and $(x_2,y_2)$ be two inputs. Then,
\[
\Pr[\Pi_{x_1y_1} \text{ reaches node } v] \times  \Pr[\Pi_{x_2y_2} \text{ reaches node } v]
= \Pr[\Pi_{x_1y_2} \text{ reaches node } v]\times \Pr[\Pi_{x_2y_1} \text{ reaches node } v].
\]
\end{lemma}

\begin{proof}
Let $u$ be an arbitrary node in the tree and $w$ be one of its children. Suppose Alice owns $u$, and the corresponding function associated with node $u$ is $a_u$.  Then
\begin{equation}  \label{eq:prob-A}
\Pr[\Pi_{xy} \text{ reaches node } w\ |\ \Pi_{xy} \text{ reaches node } u\ ]
= \Pr[a_u(x,r_A) = w],
\end{equation}
where $r_A$ is the private randomness that Alice has. Observe that this probability depends  only on $x$ and $r_A$, \emph{not} on $y$. Similarly, if Bob owns the node $u$, then this probability depends only on $y$ and $r_B$, \emph{not} on $x$. Hence, for any node $v$, 
\[
\Pr[\Pi_{xy} \text{ reaches node } v] = \alpha_v(x,r_A) \beta_v(y, r_B),
\]
where $\alpha_v(x,r_A)$ is simply the product of probabilities in the form \eqref{eq:prob-A} over all nodes that Alice owns in the unique path from the root of the tree to node $v$, and $\beta_v(y, r_B)$ is similarly given. The lemma immediately follows from this decomposition.
\end{proof}

We proceed to define information cost of a protocol. Consider the case when the input to a randomized protocol $\pi$ is also random. From the information theoretical point of view, one may ask what is the mutual information between the random transcript and the random input? 

\begin{definition}[Information cost of a protocol]   \label{def:infocost}
Let $\mu \in \Delta(\cX \times \cY)$. Let $\pi$ be a randomized protocol whose input $XY$ is from $\cX \times \cY$.  Let $\Pi$ denote the random transcript of $\pi$ when the  input $XY$ is randomly sampled according to $\mu$. That is, given a fixed transcript $t$, $\Pr[\Pi=t] = \Ex_{XY \sim \mu} \Pr[\Pi_{XY} = t]$.  The \emph{external information cost} and the \emph{internal information cost}    of protocol $\pi$ with respect to distribution $\mu$ are defined as
\begin{equation}  \label{eq:def-ICext-cost}
\IC_\mu^\ext(\pi) = I(\Pi; XY),
\end{equation}
and
\begin{equation}  \label{eq:def-IC-cost}
\IC_\mu(\pi) = I(\Pi; X |Y)+I(\Pi; Y |X).
\end{equation}
\end{definition}

We briefly introduce that one can view a randomized protocol as a random walk in the space $\Delta(\cX \times \cY)$, for detail see e.g., \cite{exactComm,my1}. Suppose without loss of generality that Alice owns the root of the tree, when Alice sends a bit $b \in \{0,1\}$, both players can update consistently from  the initial input distribution $\mu$ to the conditional distribution $\mu_b$ ($\mu$ conditioned on Alice sends $b$). In a randomized protocol, the bit $B \in \{0,1\}$ is a random bit, hence $\mu_B \in \Delta(\cX \times \cY)$ is a random distribution. Hence, a step of a randomized protocol can be equivalently viewed as a step of a random walk in the space $\Delta(\cX \times \cY)$ that goes from $\mu$ to $\mu_B$.  Repeating this process until a leaf $\ell$ of the tree is reached, let $\mu_\ell$ denote the conditional distribution of $\mu$ conditioned on the protocol reaches the leaf $\ell$. If we identify a transcript with its corresponding leaf, then $\mu_\ell = \mu | [\Pi = \ell]$.  Equivalently, $\mu_\ell \in \Delta(\cX \times \cY)$ is the final distribution reached from $\mu$ by a random walk in $\Delta(\cX \times \cY)$. 

\begin{lemma}  \label{lem:IC-over-leaves}
Let $\ell$ denote a leaf of the binary tree corresponding to a protocol $\pi$. We have, 
\[
\IC^\ext_\mu(\pi) = \sum_{\ell} \Pr[\Pi \text{ reaches }\ell] \Big(\ent_\mu(XY) - \ent_{\mu_\ell}(XY) \Big).
\]
\end{lemma}

\begin{proof}
By definition,
\begin{align*}
\IC^\ext_\mu(\pi) 
&= I(\Pi;XY) =\ent_\mu(XY) - \ent(XY|\Pi) \\
&= \ent_\mu(XY) -  \sum_{\ell}  \Pr[\Pi \text{ reaches }\ell] \ent(XY|\Pi=\ell) \\
&= \ent_\mu(XY) -  \sum_{\ell}  \Pr[\Pi \text{ reaches }\ell] \ent_{\mu_\ell}(XY) \\
&= \sum_{\ell} \Pr[\Pi \text{ reaches }\ell] \Big(\ent_\mu(XY) - \ent_{\mu_\ell}(XY) \Big). \qedhere
\end{align*}
\end{proof}
One can similarly obtain  a decomposition of $\IC_\mu(\pi)$, but we will not need it.

Finally, we define information complexities of a function. We focus on defining the various external information complexity measures of a function using the external information cost. The internal versions are similarly defined using the internal information cost. 
\begin{definition}[Information complexity of a function]   \label{def:infocomplexity}
Let $f: \cX \times \cY \to \cZ$ be a two-party function and $\pi$ be a protocol whose input comes from $\cX \times \cY$ and output is in domain $\cZ$. Let $0 \le \epsilon \le 1/2$.  

We say protocol $\pi$ computes $[f,\epsilon]$ if  $\Pr[\pi(x,y) \neq f(x,y)] \le \epsilon$ holds for \emph{every}  $(x,y) \in \cX \times \cY$. In this case, we say $\pi$ computes $f$ with \emph{point-wise} error $\epsilon$.

Let $\nu\in \Delta(\cX, \cY)$.  We say protocol $\pi$ computes $[f,\nu,\epsilon]$ if $\Ex_{(x,y) \sim \nu} \Pr[\pi(x,y) \neq f(x,y)] \le \epsilon$. In this case, we say
$\pi$ computes $f$ with \emph{$\nu$-distributional} (or simply distributional when $\nu$ is specified) error $\epsilon$. 

Let $\mu \in \Delta(\cX, \cY)$. Define the \emph{external information complexity} of $f$ with point-wise error $\epsilon$, with respect to distribution $\mu$, as,
\begin{equation}  \label{eq:def-ICext-eps}
\IC_\mu^\ext(f,\epsilon) = \inf_{\pi:\ \pi\ \text{computes}\ [f,\epsilon]} \IC_\mu^\ext(\pi).
\end{equation}

Let $\mu \in \Delta(\cX, \cY)$ and $\nu \in \Delta(\cX, \cY)$. Define the \emph{external information complexity} of $f$ with $\nu$-distributional error $\epsilon$, with respect to distribution $\mu$, as,
\begin{equation}  \label{eq:def-ICext-distri-eps}
\IC_\mu^\ext(f,\nu,\epsilon) = \inf_{\pi:\ \pi\ \text{computes}\ [f,\nu,\epsilon]} \IC_\mu^\ext(\pi).
\end{equation}

Define the \emph{prior-free external information complexity} of $f$ with point-wise error $\epsilon$ as,
\begin{equation}  \label{eq:def-ICext-pf-eps}
\IC^\ext(f,\epsilon) = \max_{\mu \in \Delta(\cX, \cY)}  \IC_\mu^\ext(f,\epsilon).
\end{equation}

Define the \emph{prior-free external information complexity} of $f$ with distributional error $\epsilon$ as,
\begin{equation}  \label{eq:def-ICext-distri-pf-eps}
\IC^{D,\ext}(f,\epsilon) = \max_{\mu \in \Delta(\cX, \cY)}  \IC_\mu^\ext(f,\mu,\epsilon). 
\end{equation}
Sometimes, we also call $\IC^{D,\ext}(f,\epsilon)$ as the distributional prior-free external information complexity of $f$ with (distributional) error $\epsilon$.

Replacing the external information cost  $\IC_\mu^\ext(\pi)$ by the internal information cost $\IC_\mu(\pi)$ in \eqref{eq:def-ICext-eps}, \eqref{eq:def-ICext-distri-eps}, \eqref{eq:def-ICext-pf-eps}, and \eqref{eq:def-ICext-distri-pf-eps}, we obtain the definitions for the corresponding internal information complexity measures of $f$: $\IC_\mu(f,\epsilon)$, $\IC_\mu(f,\nu,\epsilon)$, $\IC(f,\epsilon)$, and $\IC^{D}(f,\epsilon)$, respectively.
\end{definition}

There are two $\mu$ appearing in the definition \eqref{eq:def-ICext-distri-pf-eps}, this is \emph{not} a typo. The notation $\IC^\ext_\mu(f,\mu,\epsilon)$ means that, according to definition \eqref{eq:def-ICext-distri-eps}, we measure the information cost according to input distribution $\mu$ (see Definition \ref{def:infocost}), and we require the protocol $\pi$ to compute $[f,\mu,\epsilon]$ (i.e., $\Ex_{(x,y) \sim \mu} \Pr[\pi(x,y) \neq f(x,y)] \le \epsilon$).

When the error parameter $\epsilon = 0$, we simply write $\IC^\ext(f)$ to mean $\IC^\ext(f,0)$. This notation of suppressing $0$ error parameter will be used whenever it applies.

\subsection{Some inequalities}   \label{sec:simple}

\begin{lemma}[Pinsker,\cite{Kushilevitz1997Communication}] \label{lem:Pinsker}
$D(\mu \| \nu) \ge \frac{1}{2\ln 2} \norm{\mu-\nu}_1^2$, where $\norm{\cdot}_1$ denotes the $L_1$ norm.
\end{lemma}

\begin{lemma}[\cite{compress}] \label{lem:IC-relation}
For every distribution $\mu$ and protocol $\pi$, $\IC^\ext_\mu(\pi) \ge \IC_\mu(\pi)$.  Hence, for every $f$, $\mu$, and $\epsilon$, 
\begin{equation}   \label{eq:relation-btwn-diff-IC}
\IC^\ext_\mu(f, \epsilon)
\ge 
\begin{cases}
\IC_\mu^\ext(f,\mu,\epsilon) \\
\IC_\mu(f, \epsilon)
\end{cases}
\ge \IC_\mu(f,\mu,\epsilon).
\end{equation}
If $\mu$ is a product distribution, then $\IC^\ext_\mu(f, \epsilon) = \IC_\mu(f, \epsilon)$ and $\IC_\mu^\ext(f,\mu,\epsilon) = \IC_\mu(f,\mu,\epsilon)$.
\end{lemma}

\begin{lemma} \label{lem:entropy-small}
$h(x) \le 2\sqrt{x(1-x)}$ for $0 \le x \le 1$.
\end{lemma}

\begin{proof}
Let $g(x) = h(x) - 2\sqrt{x(1-x)}$. It is easy to verify (e.g., using Wolfram Mathematica) that there exist $0 < a < 1/2 < b < 1$ such that, $g''(x) \ge 0$ for $0 \le x \le a$ and $b \le x \le 1$, and $g''(x) \le 0$ for $a \le x \le b$. Also, $g'(1/2)= 0$.
Hence, the local maximum points of $g$ are $0,1/2$ and $1$. The claim follows as $g(0) = g(1/2) = g(1) = 0$.
\end{proof}

\begin{lemma} \label{lem:elementary}
If $m,n,r,s \ge 0$ and $mn=rs$, then $|m-r| + |m-s| \ge m-n$.
\end{lemma}

\begin{proof}
Assume $m > n$ as otherwise there is nothing to prove. By symmetry we assume $r \ge s$, then $m > s$. Either $m\ge r$ or $m < r$. If $m \ge r$, the inequality $|m-r| + |m-s| \ge m-n$ is equivalent to
$m - r - s + n 
\ge 0$. This is indeed true, since $m - r - s + n  = (m-s)(1- \frac{r}{m}) \ge 0$.
The case $m < r$ can be verified similarly.
\end{proof}

Lastly, we prove a lemma concerning divergence that will be useful in  Section \ref{sec:prod}.
By the definition of divergence, $D(\mu \| \nu) < \infty$ only if $\supp\mu \subseteq \supp\nu$. In another words, $D(\mu \| \nu)$ and $D(\nu \| \mu)$ are both finite only if $\supp\mu = \supp\nu$.

\begin{lemma}    \label{lem:Div-corrupted-upperbound}
Let $0\le \epsilon\le 1$, and suppose $D(\mu \| \nu) < \infty$ (i.e.,  $\supp \mu \subseteq \supp \nu$). Then we have
\[ 
D((1-\epsilon)\mu + \epsilon \nu \| \nu) \le  (1-\epsilon) D(\mu\|\nu) - (1 - \nu(\supp\mu)) \epsilon \log\frac{1}{\epsilon}. 
\]
\end{lemma}

\begin{proof}
As $\supp((1-\epsilon)\mu + \epsilon \nu) = \supp\mu \cup \supp\nu = \supp\nu$, we have
$D((1-\epsilon)\mu + \epsilon \nu \| \nu) < \infty$.  By convexity of $x\log x$, we have
\begin{align*}
&D((1-\epsilon)\mu + \epsilon \nu \| \nu) \\
&= \sum_{x \in \supp\mu} ((1-\epsilon)\mu(x) + \epsilon \nu(x)) \log \frac{(1-\epsilon)\mu(x) + \epsilon \nu(x)}{\nu(x)} \\
&{\phantom{=}} + \sum_{x \in (\supp\nu \backslash \supp\mu)} ((1-\epsilon)\mu(x) + \epsilon \nu(x)) \log \frac{(1-\epsilon)\mu(x) + \epsilon \nu(x)}{\nu(x)}  \\
&\le \sum_{x \in \supp\mu}  (1-\epsilon)\mu(x) \log\frac{\mu(x)}{\nu(x)}
+ \sum_{x \in \supp\mu} \epsilon \nu(x) \log \frac{\nu(x)}{\nu(x)}
+ \epsilon \log\epsilon \sum_{x \in (\supp\nu \backslash \supp\mu)} \nu(x)  \\
&= (1-\epsilon) D(\mu\|\nu) + (1 - \nu(\supp\mu)) \epsilon \log(\epsilon).
\qedhere
\end{align*}
\end{proof}

\section{Lower bounds for information cost and trivial distributions}   \label{sec:explicit-lb}
In this section we prove explicit lower bounds for both external and internal  information costs of any protocol. The lower bounds  naturally lead us to characterize\footnote{Two among four characterizations have been given in \cite{my1}. Here we treat all the characterizations using the lower bounds.}  the structures of distributions for which the information complexity vanishes (we call such distributions as \emph{trivial distributions} that will be formally defined). Both the lower bounds  and the structures of trivial distributions are useful for later study. 

We introduce a notation that will be frequently used. Given a randomized protocol $\pi$, every input $a \in \cX \times \cY$ induces a distribution over the transcripts of $\pi$. Let $p^\pi_a$ denote this distribution, i.e., for every transcript $t$, 
\begin{equation} \label{eq:def-papi}
p^\pi_a(t) = \Pr[\Pi_a = t] =\Pr[\Pi = t|XY=a].
\end{equation}
Since every transcript $t$ corresponds uniquely to a leaf $\ell$ of the binary tree, the above definition is the same as $p^\pi_a(\ell) = \Pr[\Pi_a = \ell] =\Pr[\Pi = \ell|XY=a]$. That is, $p^\pi_a$ can be also viewed as a distribution over the leaves of the binary tree for $\pi$. We will lower bound the information cost by the $L_1$ distance of two such distributions.

\subsection{The external case}   \label{sec:explicit-lb-ext}

\begin{theorem}  \label{thm:general-lb-ext-cost}
Let $\mu \in \Delta(\cX \times \cY)$ and $\pi$ be a protocol. Define $\delta^\ext = \min_{(xy) \in \supp\mu} \mu(xy)^2 > 0$. For any two inputs $a, b \in \supp\mu$, 
\[
\IC^\ext_\mu(\pi) \ge  \frac{\delta^\ext}{2\ln 2} \norm{p^\pi_a - p^\pi_b}_1^2.
\]
\end{theorem}

\begin{proof}
Let $t$ denote a transcript, let $(x,y) \in \cX \times \cY$ be an input. We use the following notation.
\begin{align}   \label{eq:general-ext-cost-Notation}
p(xy) &= \mu(xy), \notag \\
p(t) &= \Pr[\Pi=t], \notag \\
p(t,xy) &= \Pr[\Pi=t, XY=xy], \\
p(t|xy) &= \Pr[\Pi=t|XY=xy] = \Pr[\Pi_{xy} = t]]. \notag
\end{align}
Note that in our notation, $p(t|xy) = p^\pi_{xy}(t)$.

By the definition of external information cost and Pinsker inequality, 
\begin{align}  \label{eq:temp1}
\IC^\ext_\mu(\pi) 
&= I(XY;\Pi) 
= D(p(t,xy) || p(xy) p(t)) 
\ge \frac{1}{2\ln 2}  \norm{p(t,xy) - p(xy) p(t)}^2_1  \\
&=  \frac{1}{2\ln 2} \left( \sum_{xy, t} p(xy) \Big| p(t|xy) - p(t) \Big|   \right)^2 
\ge \frac{\delta^\ext}{2\ln 2} \left( \sum_{xy, t} \Big| p(t|xy) - p(t) \Big|   \right)^2. \notag
\end{align}
Now for any $a, b \in \supp\mu$, if $a=b$, then $\norm{p^\pi_a - p^\pi_b}_1 = 0$ and the lower bound in the theorem trivially holds. Assume $a=x_1y_1 \neq b=x_2y_2$, then
\begin{align}  \label{eq:temp2}
\sum_{xy, t} \Big| p(t|xy) - p(t) \Big|  
&\ge \sum_t \left( \Big| p(t|x_1y_1) - p(t) \Big|  + \Big| p(t|x_2y_2) - p(t) \Big| \right) \\
&\ge \sum_t \Big| p(t|x_1y_1) - p(t|x_2y_2)\Big| 
= \norm{p^\pi_a - p^\pi_b}_1. \notag
\end{align}
The theorem follows from \eqref{eq:temp1} and \eqref{eq:temp2}.
\end{proof}

Theorem \ref{thm:general-lb-ext-cost} immediately gives explicit lower bounds for external information complexity. We first prove a simple lemma.

\begin{lemma}  \label{lem:L1dist-explicit}
Let $f: \cX \times \cY \to \cZ$ be a function, $a, b \in \cX \times \cY$ and $f(a)\neq f(b)$. Suppose a protocol $\pi$ computes both $f(a)$ and $f(b)$ with error at most $\epsilon \in [0,1/2]$, that is, $\Pr[\Pi(a) \neq f(a)] \le \epsilon$ and $\Pr[\Pi(b) \neq f(b)] \le \epsilon$. Then,
\[
2 - 4\epsilon \le \norm{p^\pi_a - p^\pi_b}_1 \le 2.
\]
In particular, if $\pi$ computes both $f(a)$ and $f(b)$ correctly, then $\norm{p^\pi_a - p^\pi_b}_1 = 2$.
\end{lemma} 

\begin{proof}
Suppose without loss of generality $f(a)=0$ and $f(b) = 1$.
Let $\cT$ be the set of all possible transcripts of $\pi$, and $\cT_0, \cT_1$ correspond to the set of transcripts with output $0,1$, respectively. The assumption implies,
\[
\sum_{t\in \cT_0} p^\pi_a(t) = \Pr[\Pi(a) =0] \ge 1-\epsilon, \quad
\sum_{t\in \cT_0} p^\pi_b(t) = \Pr[\Pi(b) =0] \le \epsilon.
\]
Hence,
\[
\sum_{t\in \cT_0} \Big| p^\pi_a(t) - p^\pi_b(t) \Big|
\ge \sum_{t\in \cT_0} p^\pi_a(t) - \sum_{t\in \cT_0} p^\pi_b(t) \ge 1-2\epsilon.
\]
Similarly $\sum_{t\in \cT_1} \Big| p^\pi_a(t) - p^\pi_b(t) \Big| \ge 1-2\epsilon$. Hence $\norm{p^\pi_a - p^\pi_b}_1 \ge 2(1-2\epsilon)$.

The upper bound follows from the triangle inequality for the $L_1$ norm.
\end{proof}

\begin{corollary}  \label{cor:general-lb-ext-complexity}
Let $f: \cX \times \cY \to \cZ$ be a function,  $\mu$ be a distribution on $\cX \times \cY$. Let $\delta^\ext$ be defined as in Theorem \ref{thm:general-lb-ext-cost}. 
Then, 
\begin{itemize}
\item[(1)] For every $0\le \epsilon \le 1/2$, for every distribution $\mu$   such that $f$ on $\supp\mu$ is not a constant, 
\[
\IC_\mu^\ext(f,\epsilon) \ge  \frac{2(1-2\epsilon)^2\delta^\ext}{\ln 2} > 0,
\]

\item[(2)] For every distribution $\mu$  such that $f$ on $\supp\mu$ is not a  constant,  
\[
\IC_\mu^\ext(f,\mu,0) \ge  \frac{2\delta^\ext}{\ln 2} > 0.
\]
\end{itemize}
\end{corollary}

\begin{proof}
For (1), let $\pi$ be a protocol computes $[f,\epsilon]$. By the assumption, there exist $a\neq b$ both in $\supp\mu$ such that $f(a) \neq f(b)$, and $\pi$ computes both $a$ and $b$ with error at most $\epsilon$. Applying Theorem \ref{thm:general-lb-ext-cost} and Lemma \ref{lem:L1dist-explicit} proves the result.

For (2), let $\pi$ be a protocol that computes $[f,\mu,0]$. Then, $\pi$ computes every input in $\supp\mu$ correctly. Then (2) can be proved similarly as (1).
\end{proof}

\subsection{The internal case} \label{sec:explicit-lb-int}
A similar lower bound as in Theorem \ref{thm:general-lb-ext-cost}  holds for the $\IC_\mu(\pi)$, though it holds under an extra condition that we shall specify.

\subsubsection{The associated graph of a distribution}   \label{sec:assocaited-graph}
Let us think of the product set $\cX \times \cY$ as a matrix where each entry $(x,y)$ is naturally indexed by its corresponding row $x$ and corresponding column $y$. For two entries $a, b \in \cX \times \cY$, $a\neq b$, we use the notation $a\#b$ to mean $a$ and $b$ are either in the same row or in the same column. Given a distribution $\mu\in \Delta(\cX \times \cY)$, we will construct a graph $G_\mu$ whose vertex set is $\supp\mu$.

\begin{definition}[The associated graph of $\mu$, \cite{my1}]   \label{def:associated-graph}
On $\cX \times \cY$, construct a graph $G$ as follows:
\begin{itemize}
\item $V(G) = \cX \times \cY$, i.e., every point $a \in \cX \times \cY$ is a vertex in $G$;
\item $E(G)$: for every two vertices $a \neq b$, there is an edge $ab \in E(G)$ if and only if $a\#b$. 
\end{itemize}
Let $G_\mu$ be the subgraph of $G$ induced by the support of $\mu$. We call it the \emph{associated graph} of distribution $\mu$. For every connected component $C$ of $G_\mu$, define
 \begin{align*}
 C_A &= \{ x \in \cX : xy \in C \text{ for some } y \in \cY \}, \\
 C_B &= \{ y \in \cY : xy \in C \text{ for some } x \in \cX \}.
 \end{align*}
\end{definition}

Thus, $C_A \times C_B$ defines a submatrix in $\cX \times \cY$. A basic property of distinct connected components of $G_\mu$ is the following.

\begin{lemma}  \label{lem:property-connect-Gmu}
Let $C, C'$ be two distinct connected components of $G_\mu$ for a distribution $\mu$, then
\[
C_A \cap C'_A = \emptyset,\quad C_B \cap C'_B = \emptyset.
\]
\end{lemma}

\begin{proof}
Suppose for the sake of contradiction $x \in C_A \cap C'_A$. Then there are $xy \in C$, and $xy' \in C'$. By the definition of the graph $G_\mu$, there is an edge connecting $xy$ and $xy'$, contradicting to the fact that $C$ and $C'$ are distinct connected components of $G_\mu$. Similarly, one has $C_B \cap C'_B = \emptyset$.
\end{proof}
 
\begin{corollary}  \label{cor:decompose-Gmu}
Let $C_1, C_2, \ldots, C_k$ be all the distinct connected components of $G_\mu$ for a distribution $\mu$. Let $A \subseteq \cX$ be the set of $x$ where $\mu(x) > 0$, and $B \subseteq \cY$ be the set of $y$ where $\mu(y) > 0$. Then,
\[
A = \cup_{i=1}^k C_{i,\cX}, \quad B = \cup_{i=1}^k C_{i,\cY},
\]
and $C_{i,\cX}$ are pairwise disjoint, and $C_{i,\cY}$ are pairwise disjoint.
\end{corollary}

\begin{proof}
Apply directly Lemma \ref{lem:property-connect-Gmu}.
\end{proof}

Intuitively, by rearranging if necessary, Corollary \ref{cor:decompose-Gmu} means that these submatrices $C_{i,A} \times C_{i,B}$ place themselves in a diagonal fashion in $\cX \times \cY$.

\subsubsection{The lower bound}

\begin{theorem}  \label{thm:general-lb-int-cost}
Let $\mu \in \Delta(\cX \times \cY)$ and $\pi$ be a protocol. Define $\delta = \min_{(xy) \in \supp\mu} \{\frac{\mu(xy)^2}{\mu(x)}, \frac{\mu(xy)^2}{\mu(y)}\} > 0$. For any two inputs $a$ and $b$ that lie  in the same connected component of $G_\mu$,
\[
\IC_\mu(\pi) \ge  \frac{\delta}{(|\cX| + |\cY|)2\ln 2} \norm{p^\pi_a - p^\pi_b}_1^2.
\]
\end{theorem}

\begin{proof}
By definition, 
\[
\IC_\mu(\pi) = I(\Pi;Y|X) + I(\Pi;X|Y) = \sum_x \mu(x) I(\Pi;Y|X=x) + \sum_y \mu(y) I(\Pi;X|Y=y),
\] 
where $\mu(x)$ is the marginal of $\mu$ on $x$. 
Hence, those $x$ such that $\mu(x) = 0$ have no contribution to the information cost. The same applies to $y$.  Hence, without loss of generality, we assume $\mu(x), \mu(y) > 0$ hold for all $x \in \cX$ and $y \in \cY$.

Given a transcript $t$, an input $xy \in \cX \times \cY$. Let $p(t, y|x) = \Pr[\Pi=t,Y=y|X=x]$, $p(t|x) = \Pr[\Pi=t|X=x]$, and $p(y|x) = \mu(y|X=x)$. We also use the notation $p(xy)$ to mean $\mu(xy)$ when it is more convenient. One has, 
\begin{equation}  \label{eq:temp}
\mu(x) \Big( p(t,y|x) - p(t|x) p(y|x) \Big)
= p(t, xy) - p(t|x) p(xy)
= \mu(xy) \Big( p(t|xy)- p(t|x) \Big).
\end{equation}
Hence, by Pinsker inequality,
\begin{align*}
I(\Pi;Y|X) 
&= \sum_x \mu(x) I(\Pi;Y|X=x)
= \sum_x \mu(x) D(p(t,y|x) \| p(t|x) p(y|x)) \\
&\ge \sum_x  \frac{\mu(x)}{2\ln 2}  \norm{p(t,y|x) - p(t|x) p(y|x)}^2_1 \\
&= \sum_x  \frac{\mu(x)}{2\ln 2} \left( \sum_{t,y} \Big| p(t,y|x) - p(t|x) p(y|x) \Big| \right)^2 \\
&= \sum_x  \frac{1}{2\ln 2} \left( \sum_{t,y} \frac{\mu(x)}{\sqrt{\mu(x)}}\Big| p(t,y|x) - p(t|x) p(y|x) \Big| \right)^2 \\
&= \sum_x  \frac{1}{ 2\ln 2} \left( \sum_{t,y} \frac{\mu(xy)}{\sqrt{\mu(x)}} \Big| p(t|xy) - p(t|x) \Big| \right)^2 \\
&\ge \sum_x  \frac{\delta}{ 2\ln 2} d(p(t|xy),p(t|x))^2,
\end{align*}
where we used \eqref{eq:temp} in the last equality, and we used the notation $d(p(t|xy),p(t|x)) = \sum_{t,y} | p(t|xy) - p(t|x) |$.

Obviously, a similar lower bound holds for $I(\Pi;X|Y)$. Hence, 
\begin{equation}   \label{eq:lb}
\IC_\mu(\pi) 
\ge \sum_x  \frac{\delta}{ 2\ln 2} d(p(t|xy),p(t|x))^2 + \sum_y  \frac{\delta}{ 2\ln 2} d(p(t|xy),p(t|y))^2.
\end{equation}

Now, let $C$ be a connected component of $G_\mu$, and  $a, b \in C$. If $a = b$, then the lower bound in the theorem is simply $0$. Assume $a \neq b$. 

Consider firstly the simple case when $a$ and $b$ are in the same row or in the same column. Without loss of generality, assume that they are in the same row: $a=x_0 y_0, b=x_0 y_1$. Then 
\begin{align} \label{eq:lb-onerow}
d(p(t|x_0y),p(t|x_0))
&= \sum_{t,y}  \Big|p(t|x_0y) - p(t|x_0) \Big|  \notag \\
&\ge \sum_{t}  \Big| p(t|x_0y_0) - p(t|x_0) \Big| + \sum_{t}  \Big| p(t|x_0y_1) - p(t|x_0) \Big|  \\
&\ge  \sum_{t}  \Big| p(t|x_0y_0) - p(t|x_0y_1) \Big|
= \norm{p^\pi_a - p^\pi_b}_1. \notag
\end{align}
Combine this with \eqref{eq:lb} proves the theorem.

In general, $a$ and $b$ are connected by a path in $C$. Suppose the shortest path connecting $a$ and $b$ has length $k$. It is easy to see that $k \le |\cX| + |\cY|$. We will apply a telescoping argument. Indeed, let $a=a_0, b=a_k$, and suppose they are connected successively by $a_1, a_2, \ldots, a_{k-1} \in C \subseteq \supp\mu$. For every pair of consecutive points $a_i$ and $a_{i+1}$, they are either in the same row or in the same column. If they are in the same row, suppose this row corresponds to $x_i$, then by \eqref{eq:lb-onerow}, $d(p(t|x_iy),p(t|x_i)) \ge \norm{p^\pi_{a_i} - p^\pi_{a_{i+1}}}_1$. The case when they are in the same column is similar and the same lower bound holds. Observe also that every row or column contains at most two points from $a_0, a_1, a_2, \ldots, a_k$, because the path we choose has minimal length. An example is shown in Figure \ref{fig:telescope}.

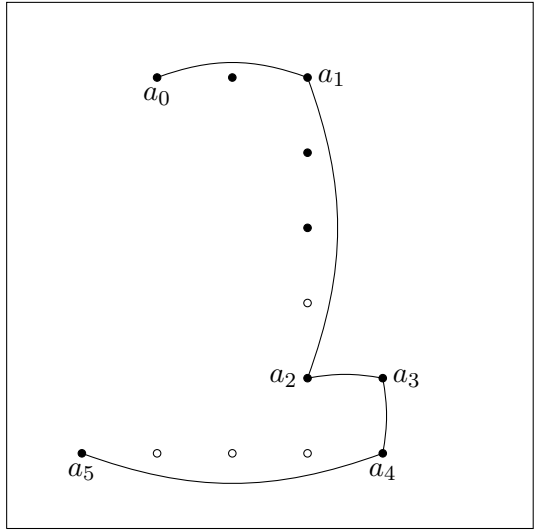
\begin{figure}[h!]   
\begin{center}
\begin{tikzpicture}
\draw (0,0) -- (7,0) -- (7,7) -- (0,7) -- (0,0);

\draw [fill] (2,6) circle [radius=0.05];
\node [below] at (2,6) {$a_0$};
\draw [fill] (3,6) circle [radius=0.05];
\draw [fill] (4,6) circle [radius=0.05];
\node [right] at (4,6) {$a_1$};
\draw [fill] (4,5) circle [radius=0.05];
\draw [fill] (4,4) circle [radius=0.05];
\draw  (4,3) circle [radius=0.05];
\draw [fill] (4,2) circle [radius=0.05];
\node [left] at (4,2) {$a_2$};
\draw [fill] (5,2) circle [radius=0.05];
\node [right] at (5,2) {$a_3$};
\draw [fill] (5,1) circle [radius=0.05];
\node [below] at (5,1) {$a_4$};
\draw (4,1) circle [radius=0.05];
\draw (3,1) circle [radius=0.05];
\draw (2,1) circle [radius=0.05];
\draw [fill] (1,1) circle [radius=0.05];
\node [below] at (1,1) {$a_5$};

\draw (2,6) to [out=20,in=160] (4,6);
\draw (4,6) to [out=290,in=70] (4,2);
\draw (4,2) to [out=10,in=170] (5,2);
\draw (5,2) to [out=280,in=80] (5,1);
\draw (5,1) to [out=200,in=340] (1,1);
\end{tikzpicture}
\end{center}
\caption{The shortest path from $a_0$ to $a_5$ in a connected component $C$, the filled dots indicate that they are in $C$.} \label{fig:telescope}
\end{figure}

Apply the lower bound from \eqref{eq:lb-onerow} to all the consecutive points in the path. By \eqref{eq:lb}, we get
\[
\IC_\mu(\pi) 
\ge \frac{\delta}{ 2\ln 2}  \sum_{i=0}^{k-1} \norm{p^\pi_{a_i} - p^\pi_{a_{i+1}}}^2_1 
\ge \frac{\delta}{ k 2\ln 2} \left( \sum_{i=0}^{k-1} \norm{p^\pi_{a_i} - p^\pi_{a_{i+1}}}_1 \right)^2 
\ge \frac{\delta}{ k 2\ln 2} \norm{p^\pi_{a_0} - p^\pi_{a_k}}_1^2,
\]
where the second inequality is by Cauchy-Schwarz, and the last one is by the triangle inequality for the $L_1$ norm.
\end{proof}

\begin{corollary}  \label{cor:general-lb-int-complexity}
Let $f: \cX \times \cY \to \cZ$ be a function, $\mu$ be a distribution on $\cX \times \cY$. Let $\delta$ be as in Theorem \ref{thm:general-lb-int-cost}. 
Then, 
\begin{itemize}
\item[(1)] For every $0\le \epsilon \le 1/2$, for every distribution $\mu$ such that  there exists at least one connected component of $G_\mu$ on which $f$ is not a constant,
\[
\IC_\mu(f,\epsilon) \ge  \frac{2(1-2\epsilon)^2\delta}{(|\cX| +|\cY|)\ln 2} > 0.
\]

\item[(2)] For every distribution $\mu$ such that  there exists at least one connected component of $G_\mu$ on which $f$ is not a constant,
\[
\IC_\mu(f,\mu,0) \ge  \frac{2\delta}{(|\cX| +|\cY|)\ln 2} > 0.
\]
\end{itemize}
\end{corollary}

\begin{proof}
Apply Theorem \ref{thm:general-lb-int-cost} and Lemma \ref{lem:L1dist-explicit}. The proof is similar to Corollary \ref{cor:general-lb-ext-complexity}.
\end{proof}

\subsection{Trivial distributions and trivial functions}  \label{sec:trivial-measures}

Obviously, computing constant functions requires no information cost. Intuitively, it seems to compute any non-constant function requires nonzero information cost. To establish this fact rigorously, we start by investigating trivial distributions.

\begin{definition}  \label{def:trivial}
Let  $f: \cX \times \cY \to \cZ$ be a function and $\mu$ be a distribution on $\cX \times \cY$. We say $\mu$ is,
\begin{itemize}
\item distributional external trivial, if $\IC_\mu^\ext(f,\mu,0) = 0$;
\item distributional internal trivial, if $\IC_\mu(f,\mu,0) = 0$;
\item external trivial, if $\IC^\ext_\mu(f) = 0$;
\item internal trivial, if $\IC_\mu(f) = 0$.
\end{itemize}
\end{definition}
 
\begin{remark}
The external trivial distributions and internal trivial distributions are already defined and characterized in \cite{my1}. However, there is a typo in \cite{my1} in defining external trivial distributions: it was defined by the condition $\IC_\mu^\ext(f,\mu,0) = 0$, but it should be  $\IC^\ext_\mu(f) = 0$. A same typo was made for defining the internal trivial distributions.
\end{remark}

The following theorem characterizes all the trivial distributions. 
 
\begin{theorem}  \label{thm:char-trivial-distri}
Let $f: \cX \times \cY \to \cZ$ be a function, $\mu$ a distribution on $\cX \times \cY$. 

The following are equivalent conditions for distributional external trivial distributions:
\begin{itemize}
\item[(1)] $\IC_\mu^\ext(f,\mu,0) = 0$; 
\item[(2)] $f$ on $\supp\mu$ is a constant;
\item[(3)] there exists a protocol $\pi$ that computes $[f,\mu,0]$ and $\IC_\mu^\ext(\pi) = 0$. 
\end{itemize}

The following are equivalent conditions for external trivial distributions:
\begin{itemize}
\item[(4)] $\IC^\ext_\mu(f) = 0$;
\item[(5)] $f$ is constant on $S_A \times S_B$, where $S_A$ is the support of the marginal of $\mu$ on Alice's input and $S_B$ is the support of the marginal of $\mu$ on Bob's input.
\item[(6)] there exists a protocol $\pi$ that computes $[f,0]$ and $\IC^\ext_\mu(\pi) = 0$.
\end{itemize}

The following are equivalent conditions for distributional internal trivial distributions:
\begin{itemize}
\item[(1')] $\IC_\mu(f,\mu,0) = 0$; 
\item[(2')] on every connected component $C$ of $G_\mu$, the restriction of $f$ on $C$ is a constant;
\item[(3')] there exists a protocol $\pi$ that computes $[f,\mu,0]$ and $\IC_\mu(\pi) = 0$. 
\end{itemize}

The following are equivalent conditions for internal trivial distributions:
\begin{itemize}
\item[(4')] $\IC_\mu(f) = 0$; 
\item[(5')] on every connected component $C$ of $G_\mu$, the restriction of $f$ on $C_A \times C_B$ is a constant;
\item[(6')] there exists a protocol $\pi$ that computes $[f,0]$ and $\IC_\mu(\pi) = 0$. 
\end{itemize}
\end{theorem}

\begin{proof}
$(1) \Longrightarrow (2)$: by Corollary \ref{cor:general-lb-ext-complexity}. 

$(2) \Longrightarrow (3)$: Suppose $f_{\supp\mu} \equiv c$. Consider the protocol $\pi$ that simply outputs $c$ without any communication. Obviously $\pi$ computes $[f,\mu,0]$ and has zero information cost.

$(3) \Longrightarrow (1)$: This is obvious by  the definition of information complexity.

$(4) \Longrightarrow (5)$: We will use Theorem \ref{thm:general-lb-ext-cost} together with the rectangle property of transcripts. Firstly, $0 \le \IC_\mu^\ext(f,\mu,0) \le \IC_\mu^\ext(f) =0$ implies $\IC_\mu^\ext(f,\mu,0) = 0$. Hence, $f$ is a constant on $\supp \mu$. Suppose $f|_{\supp\mu} \equiv z$. For the sake of contradiction suppose that (2) does not hold, i.e., there exists an input $c=xy \in S_A \times S_B$ such that $f(c) \neq z$. Since $xy \in S_A \times S_B$, there exists $x', y'$ such that $a=x'y, b=xy' \in \supp\mu$. Note that $f(a) = f(b) = z$.  Let $d = x'y'$. Let $\pi$ be any protocol that computes $[f,0]$, and $t$ be any transcript of $\pi$. The rectangle property says 
\[
p^\pi_a(t) p^\pi_b(t) = p^\pi_c(t) p^\pi_d(t).
\]
We claim that, for every transcript $t$, 
\begin{equation}  \label{eq:only-one-positive}
\begin{cases}
p^\pi_a(t) > 0 &\Longrightarrow p^\pi_b(t) = 0, \\
p^\pi_b(t) > 0 &\Longrightarrow p^\pi_a(t) = 0.
\end{cases}
\end{equation}
Indeed, since $\pi$ computes $f(a)$ correctly, $p^\pi_a(t) > 0$ implies the output of $t$ is $f(a) = z$. If $p^\pi_b(t) > 0$ at the same time, then the rectangle property implies $p^\pi_c(t) > 0$. Hence, $\Pr[\Pi(c) = z] \ge p^\pi_c(t) > 0$. However, we know $f(c) \neq z$. This contradicts the fact that $\pi$ computes $c$ correctly. Hence \eqref{eq:only-one-positive} holds. It is easy to see that \eqref{eq:only-one-positive} implies $\norm{p^\pi_a - p^\pi_b}_1 = 2$.  By Theorem \ref{thm:general-lb-ext-cost},  $\IC^\ext_\mu(f) = \inf_\pi \IC^\ext_\mu(\pi) \ge \frac{2\delta^\ext}{\ln 2} > 0$, a contradiction.

$(5) \Longrightarrow (6)$: Consider the following protocol. Alice tells Bob whether her input is in $S_A$. Bob tells Alice whether his input is in $S_B$. If the input is in $S_A \times S_B$, then the output is known. Otherwise, the players reveal their inputs (but this happens with probability zero). It is not difficult to check that this protocol has zero external information cost.

$(6) \Longrightarrow (4)$: Obvious.

$(1') \Longrightarrow (2')$: by Corollary \ref{cor:general-lb-int-complexity}.

$(2') \Longrightarrow (3')$: Let $C_1, C_2, \ldots, C_k$ be all the connected components of $G_\mu$, and suppose $f|_{C_i} \equiv d_i$. Consider the protocol $\pi$ that Alice sends $0$ if her input does not belong to the marginal of $\supp\mu$, and otherwise she sends the index $i$ to which $x \in C_{i,A}$ (note that this is possible by Corollary \ref{cor:decompose-Gmu}); and Bob does similarly. The protocol outputs $d_i$ if both Alice and Bob send $i$, and otherwise outputs an arbitrary $z \in \cZ$. Obviously $\pi$ computes $[f,\mu,0]$. Next we show $\IC_\mu(\pi)  = 0$. By definition, $\IC_\mu(\pi)  = I(X;\Pi|Y) + I(Y;\Pi|X)$. Consider $I(X;\Pi|Y) = \ent(X|Y) - \ent(X|Y\Pi)$. Observe that when $(x,y)$ is sampled according to $\mu$, from knowing $y \in C_{i,B}$ Bob knows that $x \in C_{i,A}$ holds by Corollary \ref{cor:decompose-Gmu}. As a result, $\Pi_{xy}$ is determined by $y$. Hence $\ent(X|Y\Pi) = \ent(X|Y)$, i.e., $I(X;\Pi|Y)  = 0$. By symmetry, $I(Y;\Pi|X) = 0$. Hence $\IC_\mu(\pi)= 0$.

$(3') \Longrightarrow (1')$: Obvious.

The proof for the equivalence of (4'), (5'), (6') can be obtained by combining the proof for equivalence of (4), (5), (6), and the proof for the equivalence of (1'), (2'), (3').
\end{proof}

The following is an interesting corollary for product distributions. 

\begin{corollary}  \label{cor:trivial-distri-prod-distri}
Let $f: \cX \times \cY \to \cZ$ be a function and $\mu$ a product distribution on $\cX \times \cY$. Then, the following are equivalent:
(1) $\IC^{\ext}_\mu(f,\mu,0) = 0$; (2) $\IC^{\ext}_\mu(f,0) = 0$; (3) $\IC_\mu(f,\mu,0) = 0$; (4) $\IC_\mu(f,0) = 0$;  (5) $f$ on $\supp\mu$ is a constant; (6) there exists a protocol $\pi$ that computes $[f,0]$ and $\IC^\ext_\mu(\pi) =0$. 
\end{corollary}

\begin{proof}
Since $\mu$ is a product distribution, Lemma \ref{lem:IC-relation} implies that (1) and (3) are the same, and (2) and (4) are the same. Since $\mu$ is a product distribution, it is easy to see that $\supp\mu = S_A \times S_B$ where $S_A$ and $S_B$ are as in item (5) of Theorem \ref{thm:char-trivial-distri}. Apply Theorem \ref{thm:char-trivial-distri} gives all the equivalences. 
\end{proof}

Corollary \ref{cor:trivial-distri-prod-distri} says that under a product distribution $\mu$, in computing a function with zero error, if one of the four types of information complexity measures vanishes, then they all vanish. However, if nonzero error $\epsilon > 0$ in computing a function is allowed, then $\IC_\mu(f,\mu,\epsilon) = 0$ does not imply $\IC_\mu(f,\epsilon) = 0$ even if $\mu$ is a product distribution. See the example discussed in Section \ref{sec:large-err-distri}.


We are now ready to show that only constant functions are trivial (i.e., requires no information cost to compute).

\begin{corollary}  \label{cor:trivial-function}
Let $f: \cX \times \cY \to \cZ$ be a function. Then, the following are equivalent:
(1) $\IC^{\ext}(f) = 0$; (2) $\IC^{D,\ext}(f) = 0$; (3) $\IC(f) = 0$; (4) $\IC^D(f) = 0$; (5) $f$ is a constant function.
\end{corollary}

\begin{proof}
$(1) \Longrightarrow (2) \Longrightarrow (4)$: This follows from Lemma \ref{lem:IC-relation}.

$(4) \Longrightarrow (5)$: By definition of $\IC^D(f)$, $\IC_\mu(f,\mu,0) = 0$ for every distribution $\mu$. Pick $\mu$ to be the uniform distribution, then Theorem \ref{thm:char-trivial-distri} implies  that $f$ is a constant on $\supp\mu = \cX \times \cY$.

$(5) \Longrightarrow (1)$: obvious.

We have established that (1), (2), (4), and (5) are equivalent. Similarly one can show (1), (3), (4), and (5) are equivalent. Hence, they are all equivalent.
\end{proof}

In fact, we obtain an explicit lower bound for the prior-free information complexity that does not vanish.

\begin{corollary}  \label{cor:explicit-priof-free-IC-non-zero}
Let $f: \cX \times \cY \to \cZ$ be a non-constant function. Then, 
\[
\IC^D(f) \ge \frac{1}{(|\cX|+|\cY|) \ln 16} > 0.
\]
The same lower bound holds for $\IC^{\ext}(f), \IC^{D,\ext}(f)$, and $\IC(f)$.
\end{corollary}

\begin{proof}
Since $f$ is non-constant, there exist $(x,y), (x',y') \in \cX \times \cY$ such that $f(x,y) \neq f(x',y')$. Consider the case when $x \neq x'$ and $y \neq y'$. Consider the distribution $\mu$ that is a uniform distribution on the rectangle $(x,y), (x',y), (x,y'), (x',y')$.  Apply part (2) in Corollary \ref{cor:general-lb-int-complexity} finishes the proof. The case when either $x = x'$ or $y=y'$ gives a stronger lower bound via the same argument.
\end{proof}

\section{Trading information complexity for a large error}   \label{sec:large-err}
How much information cost must be \emph{revealed} in order to compute a Boolean-valued function with an (point-wise or distributional) error at most $1/2 - \epsilon$ when $\epsilon >0$ is small? We study how the information cost depends on the parameter $\epsilon$. 

\subsection{The  point-wise error}  \label{sec:large-err-pt}
By Lemma \ref{lem:IC-relation}, the external information cost is no less than the internal information cost. Hence, we will prove an upper bound for the external information complexity and a lower bound for the internal information complexity.

\begin{theorem} \label{thm:large-err-pt-ub}
For every $f: \cX \times \cY \to \{0,1\}$, every distribution $\mu \in \Delta(\cX \times \cY)$, and every $0\le \epsilon \le 1/2$, we have $\IC_\mu^\ext\left(f, 1/2-\epsilon\right) \le 2\epsilon \IC_\mu^\ext(f,0)$.
Similarly, $\IC_\mu(f,1/2-\epsilon) \le 2\epsilon \IC_\mu(f,0)$.
\end{theorem}

\begin{proof}
Consider a protocol $\pi$ in which with probability $1-2\epsilon$, Alice and Bob simply output $0$ or $1$ with equal probability,  and with probability $2\epsilon$ Alice and Bob run a protocol that computes $[f,0]$ and has optimal (or near-optimal) external information cost. Obviously, protocol $\pi$ has error at most $\frac{1-2\epsilon}{2} = 1/2 - \epsilon$, and has external information cost $\IC_\mu^\ext(\Pi) \le 2\epsilon \IC_\mu^\ext(f,0)$.
\end{proof}

We proceed to show a lower bound. To simplify the statement of the theorem we introduce a terminology. Given a function $f: \cX \times \cY \to \cZ$ and a distribution $\mu$ on $\cX \times \cY$, let $C$ be a connected component of $G_\mu$ (see Definition \ref{def:associated-graph}). We say  \emph{$C$ contains an $\AND$ block}  if its corresponding matrix $C_A \times C_B$ contains a submatrix $\begin{pmatrix}
a & b \\
c & d
\end{pmatrix}$ such that $a,b,c\in C \subseteq \supp\mu$, but $d\not\in C$, and $f(a)=f(b)=f(c)\neq f(d)$.

\begin{theorem}   \label{thm:large-err-pt-lb}
Let $f: \cX \times \cY \to \{0,1\}$ and distribution $\mu$ be such that $\IC_\mu(f,0) > 0$. Let $\delta > 0$ be as defined in Theorem \ref{thm:general-lb-int-cost}. For every $0\le \epsilon \le 1/2$, if either of the following two conditions is satisfied:
\begin{enumerate}[(1)]
\item there is a connected component $C$ of $G_\mu$ such that $f$ is not a constant on $C$, 

\item  there is a connected component $C$ of $G_\mu$ such that $C$ contains an $\AND$ block,
\end{enumerate}
then 
\[
\IC_\mu(f,1/2-\epsilon) \ge \frac{\delta}{(|\cX| + |\cY|)2\ln 2} \epsilon^2.
\]
\end{theorem}

\begin{proof}
Let $\pi$ be an arbitrary protocol that computes $[f,1/2-\epsilon]$. 
By Theorem~\ref{thm:general-lb-int-cost}, it suffices to find an $\Omega(\epsilon)$ lower bound for the $L_1$ distance between the distribution $p^\pi_a$ and $p^\pi_b$ for two points $a$ and $b$ in some connected component of $G_\mu$. By Theorem \ref{thm:char-trivial-distri}, $\IC_\mu(f,0) > 0$ implies  the existence of a connected component $C$ of $G_\mu$ such that $f$ is not constant on $C_A \times C_B$.

{\noindent \textbf{Case 1}}: $f$ is not constant on $C$. Hence there exists $a, b \in C$ such that $f(a) \neq f(b)$. Since $\pi$ computes $[f, 1/2-\epsilon]$, by Lemma \ref{lem:L1dist-explicit}, $\norm{p^\pi_a - p^\pi_b}_1 \ge 2 - 4(1/2 - \epsilon) = 4\epsilon$. 

{\noindent \textbf{Case 2}}: $f$ is constant on $C$ and $C$ contains an $\AND$ block. Without loss of generality assume $f=0$ on $C$. Let the $\AND$ block be
$\begin{pmatrix}
a & b \\
c & d
\end{pmatrix}$.
Then, $f(a)=f(b)=f(c) = 0$ and $f(d) = 1$. Note that $d \not\in C$. For every transcript $t$, the rectangle property Lemma \ref{lem:rectangle-property} says $p_a^\pi(t)p_d^\pi(t) = p_b^\pi(t)p_c^\pi(t)$. Let $T$ be the set of transcripts that output $0$. By Lemma~\ref{lem:elementary}, 
\begin{align*}
\norm{p^\pi_a - p^\pi_b}_1 + \norm{p^\pi_a - p^\pi_c}_1 
&\ge \sum_{t\in T} \Big( |p_a^\pi(t) - p_b^\pi(t)| + |p_a^\pi(t) - p_c^\pi(t)| \Big) \\
&\ge \sum_{t\in T} \Big(p_a^\pi(t) - p_d^\pi(t) \Big) 
=\sum_{t\in T} p_a^\pi(t) - \sum_{t\in T}  p_d^\pi(t)
\ge 2\epsilon.
\end{align*}
Hence either $\norm{p^\pi_a - p^\pi_b}_1$ or $\norm{p^\pi_a - p^\pi_c}_1$ is bounded below by $\epsilon$.
\end{proof}

\subsection{The distributional error}  \label{sec:large-err-distri}
A simple adaptation of the proof for Theorem \ref{thm:large-err-pt-ub} implies the following.

\begin{corollary}  \label{thm:large-err-distri-ub}
For every $f: \cX \times \cY \to \{0,1\}$, every distribution $\mu \in \Delta(\cX \times \cY)$, and every $0\le \epsilon \le 1/2$, we have $\IC_\mu^\ext(f,\mu,1/2-\epsilon) \le 2\epsilon \IC_\mu^\ext(f,\mu,0)$. Similarly, $\IC_\mu(f,\mu,1/2-\epsilon) \le 2\epsilon \IC_\mu(f,\mu,0)$
\end{corollary}

It turns out that no general lower bound exists when distributional error is allowed. Consider the $\AND$ function and let $\mu$ be the uniform distribution. By Theorem \ref{thm:char-trivial-distri}, $\IC_\mu(\AND,\mu,0) > 0$. However, $\IC_\mu(\AND,\mu, 1/4) = 0$ via the following protocol $\pi$: simply output $0$ and terminate. Obviously, $\pi$ has zero information cost and it computes $[\AND,\mu,1/4]$. Note that $\IC_\mu(\AND,1/4) \ge \frac{1}{1024\ln 2} > 0$ by Theorem \ref{thm:large-err-pt-lb}.

\subsection{The prior-free information complexity}   \label{sec:pf-large-err}

\begin{corollary}  \label{cor:pf-large-err-pt}
For every non-constant $f: \cX \times \cY \to \{0,1\}$, and every $0\le \epsilon \le 1/2$, 
\[
\frac{1}{2\ln 2}\cdot \frac{1}{(|\cX|+|\cY|) \cdot |\cX \times \cY| \cdot \max\{|\cX|, |\cY|\}} \epsilon^2 
\le \IC(f, 1/2 - \epsilon) \le \IC^\ext(f, 1/2 - \epsilon) \le 2\epsilon \IC^\ext(f,0).
\]
\end{corollary}

\begin{proof}
The upper bound follows from Theorem \ref{thm:large-err-pt-ub}. The lower bound is obtained by applying uniform distribution with  Theorem \ref{thm:large-err-pt-lb}.
\end{proof}

As Corollary \ref{cor:pf-large-err-pt} holds for \emph{every} $0 \le \epsilon \le 1/2$, this gives bounds on information complexity with \emph{any} error.  Of course, the upper and lower bounds obtained in this way are weak. One often can obtain better bounds. One such example is $0.5  \le \IC^\ext(\XOR,1/4) \le 1$, a much better lower bound from Theorem \ref{thm:XOR-ext-err}.

\section{Trading  external information complexity for a small  error}  \label{sec:ext-small-err}
Comparing to computing a function $f$ without error, how much external information cost one can save by allowing a small (point-wise or distributional) error $\epsilon>0$?

\subsection{The point-wise error}  \label{sec:ext-pt-err}
We will study the behaviour of  $\IC_\mu^\ext(f,0) - \IC_\mu^\ext(f,\epsilon)$ with respect to $\epsilon$ when $\epsilon > 0$ is small. 

Consider the upper bound first. Note that $\IC_\mu^\ext(f,0)  - \IC_\mu^\ext(f,\epsilon) \le w$ is equivalent to $\IC_\mu^\ext(f,0) \le \IC_\mu^\ext(f,\epsilon) + w$, i.e., it gives an upper bound for $\IC_\mu^\ext(f,0)$ in terms of  $\IC_\mu^\ext(f,\epsilon)$. This suggests the following: given a protocol that computes $[f,\epsilon]$, can we modify  it  such that it computes $[f,0]$? This idea is called the ``protocol completion'' method, and has been used in \cite{exactComm}. In \cite{my1}, a conceptually simpler ``protocol completion'' has been introduced to study the internal information complexity. It turns out that the proof in \cite{my1} for the internal information complexity also works for the external information complexity. 

\begin{theorem}  \label{thm:ext-ub-small-pt-err}
For every $f: \cX \times \cY \to \cZ$, every distribution $\mu$, and $0\le \epsilon \le 1/4$, we have
\[
\IC_\mu^\ext(f,0)  - \IC_\mu^\ext(f,\epsilon) \le 4|\cX\times \cY|h(\sqrt{\epsilon}).
\]
\end{theorem}

See Figure \ref{fig:protocol-completion} for the protocol completion algorithm for Theorem \ref{thm:ext-ub-small-pt-err}: the protocol $\pi$ is assumed to compute $[f,\epsilon]$, and the protocol $\pi'$, obtained by protocol completion from $\pi$, computes $[f,0]$. The notation $\mu_\ell$ means the same as in Lemma \ref{lem:IC-over-leaves}. For the proof, see \cite[Theorem 3.5]{my1} (one only needs to make small modification of the proof to adapt external information complexity).

\begin{figure}[ht!]
\begin{framed}
On input $(X,Y)$:
\begin{itemize}
\item[(1)] Alice and Bob run the protocol $\pi$ with input $(X,Y)$ and reach a leaf $\ell$, let $z_\ell$ denote the output at leaf $\ell$; 
\item[(2)] Let $\Omega_\ell = \{(x,y): f(x,y)\neq z_\ell\}$. Alice and Bob verify whether $(X,Y) \in \Omega_\ell$ as follows: for every $(x,y) \in \Omega_\ell$ they verify whether $(X,Y) = (x,y)$ as follows,  
\begin{itemize}
\item If $\mu_\ell(x) \le \mu_\ell(y)$,  Alice reveals whether $X = x$ to Bob, and if yes, Bob reveals whether $Y=y$ to Alice.

\item If $\mu_\ell(x) > \mu_\ell(y)$,  Bob initiates the verification process.
\end{itemize}

If $(X,Y)=(x,y)$, they output $f(x,y)$ and terminate. 

Otherwise, they go to the next element in $\Omega_\ell$ and repeat. 

If they find out $(X,Y) \not\in \Omega_\ell$, they output $z_\ell$ and terminate.

\end{itemize}
\end{framed}
\caption{The protocol $\pi'$ via protocol completion from $\pi$, from \cite{my1}.}   \label{fig:protocol-completion}
\end{figure}

Next we consider the lower bound. One can easily save  $\Omega(\epsilon)$ external information cost. 

\begin{theorem}   \label{thm:ext-trivial-lb-small-pt-err}
For every $f: \cX \times \cY \to \cZ$ and distribution $\mu$ such that $\IC_\mu^\ext(f,0) > 0$, for every $0\le \epsilon \le 1/2$, we have 
\[
\IC_\mu^\ext(f,0) - \IC_\mu^\ext(f,\epsilon) \ge  \epsilon \IC_\mu^\ext(f,0).
\]
If the function $f$ is Boolean-valued: $f: \cX \times \cY \to \{0,1\}$,  then,
\[
\IC_\mu^\ext(f,0) - \IC_\mu^\ext(f,\epsilon) \ge  2\epsilon \IC_\mu^\ext(f,0).
\]
\end{theorem}

\begin{proof}
Given an arbitrary $\delta > 0$, by definition of $\IC_\mu^\ext(f,0)$,  there exists a protocol $\pi$ that computes $[f,0]$ and has $\IC_\mu^\ext(\pi) \le \IC_\mu^\ext(f,0) + \delta$. Now consider a protocol $\pi'$ that with probability $1-\epsilon$ runs $\pi$, and it outputs randomly otherwise. Obviously $\pi'$ computes $[f,\epsilon]$, and has external information cost 
\[
\IC_\mu^\ext(\pi') = (1-\epsilon)\IC_\mu^\ext(\pi)
\le (1-\epsilon)\IC_\mu^\ext(f,0) + (1-\epsilon)\delta.
\]
Since $\delta$ is arbitrary, one must have $\IC_\mu^\ext(f,\epsilon) \le (1-\epsilon)\IC_\mu^\ext(f,0)$.

When $\cZ = \{0,1\}$, consider a similar protocol $\pi'$ that with probability $1-2\epsilon$ it runs $\pi$, and otherwise outputs $0$ or $1$ uniformly at random. Then the error that $\pi'$ makes is at most  $2\epsilon \cdot 1/2= \epsilon$. Hence $\pi'$ computes $[f,\epsilon]$. The rest of the proof is identical.
\end{proof}

In general, the linear dependency on $\epsilon$ can not be improved. An example has been provided in \cite[Proposition 3.4]{my1}. We will give a simpler proof for that example in Section \ref{sec:XOR} with a better bound. In fact, our proof gives an exact \emph{equality}, not just an upper bound. 

\subsubsection{Product distributions}  \label{sec:prod}

By Lemma \ref{lem:IC-relation}, when $\mu$ is a product distribution the external and internal information complexity are the same. In \cite[Theorem 3.2]{my1}, an order $h(\epsilon)$ lower bound has been shown for internal information complexity (for all internal non-trivial distributions). The proof, however, is relatively complicated. Following essentially the same idea as for \cite[Theorem 3.2]{my1}, we provide a much simpler proof for product distributions using Lemma \ref{lem:Div-corrupted-upperbound}.

\begin{theorem}   \label{thm:ext-prod-lb-small-pt-err}
Let $f\colon\cX \times \cY \to \cZ$ and  $\mu=\mu_1 \times \mu_2$ be a product distribution on $\cX \times \cY$, and suppose $\IC_\mu^\ext(f,0) > 0$. Then, for  every $0\le \epsilon \le 1/2$, 
\[ 
\IC_\mu^\ext(f,0) - \IC_\mu^\ext(f,\epsilon) \ge \frac{1 - \sqrt{1-\delta}}{4} h(\epsilon) + \frac{\epsilon}{2} \IC_\mu^\ext(f,0)
\]
where $\delta = \min_{z\in \cZ:\ f^{-1}(z)\cap \supp\mu \neq \emptyset} \mu(f^{-1}(z)\cap \supp\mu)$, and $0 < \delta < 1$.
\end{theorem}

The idea can be called as ``private cheating'':  one player privately cheats in the communication by randomly deciding not to use her/his real input. Consider Alice for instance. Instead of using her real input $x$ in the communication with Bob, she cheats in the communication by using some random input $x'$ and sends (probably erroneous) bits accordingly. If Alice  only cheats with small probability, say $\epsilon$, and otherwise still uses her real input, then the protocol will only make a small error. As we shall see,  this reduces the information cost by  $\Omega(h(\epsilon))$.

\begin{proof}
Since $\IC_\mu^\ext(f,0) > 0$ and $\mu$ is a product distribution, Corollary \ref{cor:trivial-distri-prod-distri} implies that $f$ on $\supp\mu$ is not a constant. Hence $0< \delta < 1$.

Let $\pi$ be a protocol that computes $[f,0]$. Define a protocol $\pi'$ as in Figure \ref{fig:protocol-cheat-prod-distri}.

\begin{figure}[h!]
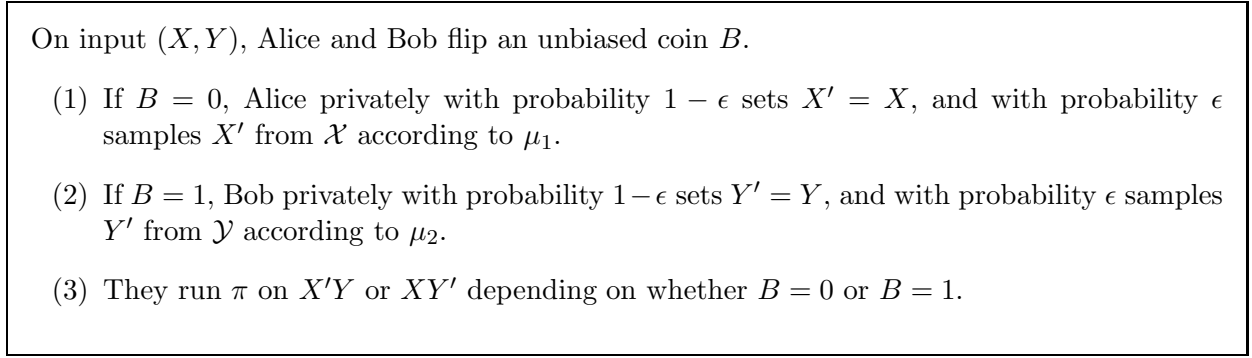

\begin{framed}
On input $(X,Y)$, Alice and Bob flip an unbiased coin $B$.
\begin{itemize}
\item[(1)] If $B=0$, Alice privately with probability $1-\epsilon$ sets $X'=X$, and with probability $\epsilon$ samples $X'$ from $\cX$ according to $\mu_1$.
\item[(2)] If $B=1$, Bob privately with probability $1-\epsilon$ sets $Y'=Y$, and with probability $\epsilon$ samples $Y'$ from $\cY$ according to $\mu_2$.
\item[(3)] They run $\pi$ on $X'Y$ or $XY'$ depending on whether $B=0$ or $B=1$.
\end{itemize}
\end{framed}
\caption{The protocol $\pi'$ via private cheating from $\pi$.}   \label{fig:protocol-cheat-prod-distri}
\end{figure}

Obviously $\pi'$ computes $[f,\epsilon]$, i.e., it  computes $f(x,y)$ correctly with probability at least $1-\epsilon$ for every $xy \in \cX \times \cY$. Let $\pi_0'$ and $\pi_1'$ be the above protocol restricted to $B=0$ and $B=1$, respectively. Clearly
\begin{equation}  \label{eq:Rate-Prod-distribution-temp1}
\IC_\mu^\ext(\pi')=\frac{\IC_\mu^\ext(\pi'_0)+\IC_\mu^\ext(\pi'_1)}{2}.
\end{equation}

Let $\Pi$ denote the random transcript of $\pi$. Let us focus on $\IC_\mu^\ext(\pi'_0)$. Note that $X'Y$ has the same distribution as $XY$. Hence, as $\mu$ is a product distribution,
\[
\IC_\mu^\ext(\pi'_0) = I(\Pi_{X'Y};X) +  I(\Pi_{X'Y};Y)=I(\Pi_{X'Y};X) +  I(\Pi_{XY};Y).
\]
By the definition of $X'$, we have
\[  
\Pr[X'=a' |  X=a]
=
\begin{cases}
1 - \epsilon + \epsilon \mu_1(a), &\quad a' = a, \\
\epsilon \mu_1(a'), &\quad a' \neq a.
\end{cases}
\]
Since $X$ and $X'$ have the same distribution, by Bayes' rule, we have
\begin{align*}
\Pr[X=a |  X'=a']
&= \frac{\Pr[X'=a' |  X=a] \Pr[X=a]}{\Pr[X'=a']} = \frac{\mu_1(a)}{\mu_1(a')} \Pr[X'=a' |  X=a]\\
&=
\begin{cases}
1 - \epsilon + \epsilon \mu_1(a), &\quad a' = a, \\
\epsilon \mu_1(a), &\quad a' \neq a.
\end{cases}
\end{align*}
Now for $a \in \cX$ and a fixed transcript $t$,
\begin{align*}
\Pr[X=a| \Pi_{X'Y}=t]
&= \sum_{a' \in \cX } \Pr[X=a,X'=a'| \Pi_{X'Y}=t] \\
&=  \sum_{a' \in \cX } \Pr[X'=a'| \Pi_{X'Y}=t]  \Pr[X=a |  X'=a', \Pi_{X'Y}=t]  \\
&= \sum_{a' \in \cX } \Pr[X=a'| \Pi_{XY}=t]  \Pr[X=a |  X'=a'] \\
&= \Pr[X=a | \Pi_{XY}=t] (1 - \epsilon + \epsilon \mu_1(a)) \\
&{\phantom{=}} +  \sum_{a' \in \cX, a' \neq a} \Pr[X=a'| \Pi_{XY}=t] \epsilon \mu_1(a)  \\
&= (1-\epsilon) \Pr[X=a | \Pi_{XY}=t]
  + \epsilon \mu_1 (a) \sum_{a' \in \cX}  \Pr[X=a'| \Pi_{XY}=t] \\
&= (1-\epsilon) \Pr[X=a | \Pi_{XY}=t] + \epsilon \mu_1 (a).
\end{align*}
Denote the distribution of $X|_{\Pi_{XY} = t}$ by $\mu_{t,X}$, and $X|_{\Pi_{X'Y} = t}$ by $\mu'_{t,X}$, then the above formula says
$$
\mu'_{t,X} = (1-\epsilon) \mu_{t, X} + \epsilon \mu_1.
$$
Since the distribution of $X'Y$ is the same as $XY$, by Lemma \ref{lem:Div-corrupted-upperbound}, 
\begin{align*}
I(\Pi_{X'Y}; X)
&= \Ex_{t \sim \Pi_{X'Y}} D(X|_{\Pi_{X'Y} = t} \|X) \\
&= \Ex_{t \sim \Pi_{XY}} D(\mu'_{t,X}\|\mu_1) = \Ex_{t \sim \Pi_{XY}} D((1-\epsilon) \mu_{t, X} + \epsilon \mu_1\|\mu_1)  \\
&\le
\Ex_{t \sim \Pi_{XY}} \big( (1-\epsilon) D(\mu_{t, X} \| \mu_1) - (1-\mu_1(\supp\mu_{t, X})) \epsilon\log\frac{1}{\epsilon} \big) \\
&= (1-\epsilon)  \Ex_{t \sim \Pi_{XY}} D(X|_{\Pi_{XY} = t} \|X) - \Ex_{t \sim \Pi_{XY}} (1-\mu_1(\supp\mu_{t, X})) \epsilon\log\frac{1}{\epsilon}   \\
&= (1-\epsilon)I(\Pi_{XY}; X) - \Ex_{t \sim \Pi_{XY}} (1-\mu_1(\supp\mu_{t, X}))\epsilon\log\frac{1}{\epsilon}.
\end{align*}
Hence, 
\begin{equation} \label{eq:Rate-Prod-distribution-temp2}
\IC_\mu^\ext(\pi'_0)
\le
I(\Pi_{XY};Y) + (1-\epsilon)I(\Pi_{XY}; X) - \left( \epsilon\log\frac{1}{\epsilon} \right) \Ex_{t \sim \Pi_{XY}} (1-\mu_1(\supp\mu_{t, X})).
\end{equation}
Similarly, one has
\begin{equation}\label{eq:Rate-Prod-distribution-temp3}
\IC_\mu^\ext(\pi'_1)
\le
I(\Pi_{XY};X) + (1-\epsilon)I(\Pi_{XY}; Y) - \left( \epsilon\log\frac{1}{\epsilon} \right) \Ex_{t \sim \Pi_{XY}} (1-\mu_2(\supp\mu_{t, Y})),
\end{equation}
where $\mu_{t, Y}$ denotes the distribution of $Y|_{\Pi_{XY}=t}$. 

Let $\mu_t$ be the distribution of $XY|_{\Pi_{XY}=t}$. If $\mu$ is a product distribution, then $\mu_t$ is also a product distribution, and $\supp\mu_t \subseteq \supp\mu$ (see, e.g., \cite[Section 2.5]{my1}). Hence, $\mu_t = \mu_{t,X} \times \mu_{t, Y}$. Let $z_t$ denote the output of transcript $t$. Note that for every $(x,y) \in \supp\mu_t$, the transcript $t$ outputs $z_t$ for this $(x,y)$. Since protocol $\pi$ computes $f$ correctly on every input,  one must have $\supp\mu_t \subseteq f^{-1}(z_t)$. Hence,  $\supp\mu_t \subseteq f^{-1}(z_t) \cap \supp\mu$. By the definition of $\delta$ and the fact that $f$ on $\supp\mu$ is \emph{not} a constant, one has $\mu(\supp\mu_t) \le 1 - \delta$.
As both $\mu$ and $\mu_t$ are product distributions, we have $
\mu(\supp\mu_t) = \mu(\supp\mu_{t,X} \times \supp\mu_{t,Y}) = \mu_1(\supp\mu_{t,X}) \times \mu_2(\supp\mu_{t,Y})$.
Therefore, $\mu_1(\supp\mu_{t,X}) \times \mu_2(\supp\mu_{t,Y}) \le 1 - \delta$.
Hence $\min\{\mu_1(\supp\mu_{t,X}), \mu_2(\supp\mu_{t,Y})\} \le \sqrt{1-\delta}$, implying that
\begin{equation} \label{eq:Rate-Prod-distribution-temp4}
(1 - \mu_1(\supp\mu_{t,X})) + (1 - \mu_2(\supp\mu_{t,Y})) \ge 1 - \sqrt{1-\delta}.
\end{equation}
By \eqref{eq:Rate-Prod-distribution-temp1}, \eqref{eq:Rate-Prod-distribution-temp2}, \eqref{eq:Rate-Prod-distribution-temp3} and \eqref{eq:Rate-Prod-distribution-temp4}, we get the desired bound by applying the fact that $\epsilon \log (1/\epsilon) \ge h(\epsilon)/2$ for all $0 \le \epsilon \le 1/2$.
\end{proof}

\subsection{The distributional error}   \label{sec:ext-distri-err}

Similar bounds can be obtained for  $\IC_\mu^\ext(f,\mu, 0) - \IC_\mu^\ext(f,\mu, \epsilon)$, in the same way as in Section \ref{sec:ext-pt-err}.

\begin{theorem}   \label{thm:ext-small-distri-err}
For every $f: \cX \times \cY \to \cZ$, every distribution $\mu$ such that $\IC^\ext_\mu(f,\mu, 0) > 0$, we have the following,
\begin{enumerate}[(1)]
\item For every $0\le \epsilon \le \alpha/4$,  $\IC_\mu^\ext(f,\mu, 0)  - \IC_\mu^\ext(f,\mu, \epsilon) \le 4|\cX \times \cY|h(\sqrt{\epsilon/\alpha})$ where $\alpha = \min_{(x,y) \in \supp\mu} \mu(x,y) > 0$.

\item For every $0 \le \epsilon \le 1/2$, $\IC_\mu^\ext(f,\mu, 0)  - \IC_\mu^\ext(f,\mu, \epsilon) \ge  \epsilon \IC_\mu^\ext(f,\mu, 0) = \Omega(\epsilon)$. 
If the function $f$ is Boolean-valued: $f: \cX \times \cY \to \{0,1\}$,  then $\IC_\mu^\ext(f,\mu, 0)  - \IC_\mu^\ext(f,\mu, \epsilon) \ge  2\epsilon \IC_\mu^\ext(f,\mu, 0)$.
\end{enumerate}
\end{theorem}

\begin{proof}
The proof for the upper bound is similar to the proof of Theorem~\ref{thm:ext-ub-small-pt-err}.  Consider a protocol $\pi$  that computes $[f,\mu, \epsilon]$.  The new protocol  $\pi'$ that computes $[f, \mu, 0]$ is defined similarly as in Figure \ref{fig:protocol-completion}, the difference is that the verification is only performed on the support of $\mu$, i.e., $\Omega'_\ell = \{(x,y) : f(x,y)\neq z_\ell\} \cap \supp\mu$.
Note that $\pi$ has point-wise error at most $\epsilon/\alpha$ on every input in $\supp\mu$. Thus the same analysis in the proof of Theorem~\ref{thm:ext-ub-small-pt-err} gives the upper bound. 
The proof for the lower bound is the same as for Theorem \ref{thm:ext-trivial-lb-small-pt-err}.
\end{proof}

\subsection{The prior-free external information complexity}   \label{sec:ext-pf}

By Corollary \ref{cor:trivial-function}, $\IC^\ext(f,0) > 0$, or equivalently $\IC^{D,\ext}(f,0) > 0$, if and only if $f$ is not a constant function. The following results are direct consequences of results in Section \ref{sec:ext-pt-err} and Section \ref{sec:ext-distri-err}. The explicit lower bounds are obtained by applying Corollary \ref{cor:explicit-priof-free-IC-non-zero}.

\begin{corollary}  \label{cor:pf-ext-pt-err}
Let $f: \cX \times \cY \to \cZ$ be a non-constant function. Then, 
\begin{enumerate}[(1)]
\item  For every $0\le \epsilon \le 1/4$, $\IC^\ext(f,0)  - \IC^\ext(f,\epsilon) \le 4|\cX \times \cY|h(\sqrt{\epsilon})$.

\item For every $0\le \epsilon \le 1/2$, $\IC^\ext(f,0) - \IC^\ext(f,\epsilon) \ge  \epsilon \IC^\ext(f,0) \ge \frac{\epsilon}{(|\cX|+|\cY|) \ln 16}$. 
If the function $f$ is Boolean-valued: $f: \cX \times \cY \to \{0,1\}$,  then, $\IC^\ext(f,0) - \IC^\ext(f,\epsilon) \ge  2\epsilon \IC^\ext(f,0)\ge \frac{\epsilon}{(|\cX|+|\cY|) \ln 4}$.
\end{enumerate}
\end{corollary}

\begin{corollary}   \label{cor:pf-ext-distri-err}
Let $f: \cX \times \cY \to \cZ$ be a non-constant function. Then, 
\begin{enumerate}[(1)]
\item There exists a constant $\alpha > 0$ that depends only on $f$, such that for every $0\le \epsilon \le \alpha/4$, 
$\IC^{D,\ext}(f,0)  - \IC^{D,\ext}(f, \epsilon) \le 4|\cX \times \cY|h(\sqrt{\epsilon/\alpha})$.

\item For every $0\le \epsilon \le 1/2$, $\IC^{D,\ext}(f,0) - \IC^{D,\ext}(f,\epsilon) \ge  \epsilon \IC^{D,\ext}(f,0) \ge \frac{\epsilon}{(|\cX|+|\cY|) \ln 16}$.
If the function $f$ is Boolean-valued: $f: \cX \times \cY \to \{0,1\}$,  then, $\IC^{D,\ext}(f,0) - \IC^{D,\ext}(f,\epsilon) \ge  2\epsilon \IC^{D,\ext}(f,0) \ge \frac{\epsilon}{(|\cX|+|\cY|) \ln 4}$.
\end{enumerate}
\end{corollary}

\section{Tight examples}   \label{sec:eg}

We say a result is tight if the dependency on the parameter $\epsilon$ cannot be improved with respect to the order of $\epsilon$. We will see that Theorem \ref{thm:large-err-pt-ub}, Theorem \ref{thm:large-err-pt-lb}, and Theorem \ref{thm:ext-trivial-lb-small-pt-err} are all tight.

Let $\DISJ_n: \{0,1\}^n \times \{0,1\}^n \to \{0,1\}$ be the two-party $n$-bit set disjointness function: $\DISJ_n(x,y)=1$ if and only if $x$ and $y$, when viewed as subsets of $\{1,2,\ldots,n\}$, are disjoint. In \cite{twobounds}, it is shown that for $n = 3k$, there exists a distribution $\mu$ on the input space $\{0,1\}^n \times \{0,1\}^n$, such that for sufficiently small $\epsilon > 0$, $\IC_\mu(\DISJ_n, 1/2-\epsilon) = \Omega(\epsilon n) = \Omega(\epsilon \cdot \IC_\mu(\DISJ_n,0))$. This shows that Theorem \ref{thm:large-err-pt-ub} is  tight. The distribution is constructed as follows: let $\mu_0$ be the uniform distribution on the following six pairs: $(100,010), (100,001), (010,100), (010,001), (001,100), (001,010)$. That is, every pair consists of two disjoint subsets each of size one. Then $\mu$ is the product distribution $\mu = \mu_0 \times \mu_0 \times \cdots \mu_0$ ($k$-times).

Next we will see that Theorem \ref{thm:ext-trivial-lb-small-pt-err} and Theorem \ref{thm:large-err-pt-lb} are also tight.

\subsection{The $\XOR$ example}   \label{sec:XOR}

\begin{theorem} \label{thm:XOR-ext-err}
Let 
$
 \mu = 
\begin{pmatrix}
1/2 & 0 \\
0 & 1/2
\end{pmatrix}$ 
be a distribution on $\{0,1\} \times \{0,1\}$. 
Then, for every $0 \le \epsilon \le 1/2$, 
\begin{equation}   \label{eq:XOR-eps}
\IC_\mu^\ext(\XOR,\epsilon)  = 1 - 2\epsilon.
\end{equation}
\end{theorem}

Theorem \ref{thm:XOR-ext-err} implies that Theorem \ref{thm:ext-trivial-lb-small-pt-err}  is tight, it also implies Theorem \ref{thm:large-err-pt-ub} is tight for the external information complexity. This example has been analyzed in \cite[Proposition 3.4]{my1} where $\IC_\mu^\ext(\XOR,\epsilon)  \ge 1 - 3\epsilon$ was shown. We improve it to the \emph{equality} \eqref{eq:XOR-eps}.  This is interesting since, to the knowledge of the author, \eqref{eq:XOR-eps} is the first non-trivial example where we know the \emph{exact} external information complexity \emph{with a non-zero error} for an explicit function. For internal information complexity, it seems we do not know such an example. The proof in \cite[Proposition 3.4]{my1} uses a result from real analysis. We replace that by a simple inequality Lemma \ref{lem:entropy-small}. For reader's convenience, we provide a complete proof. 

\begin{proof}
By Theorem \ref{thm:char-trivial-distri}, $\IC_\mu^\ext(\XOR,0) > 0$. Obviously, $\IC_\mu^\ext(\XOR,0) \le h(1/2) = 1$. We claim that it suffices to show
\begin{equation}  \label{eq:xor-goal}
\IC_\mu^\ext(\XOR,\epsilon) \ge 1-2\epsilon. 
\end{equation}
Indeed, firstly this implies $\IC_\mu^\ext(\XOR,0) =1$. Furthermore, it also implies $\IC_\mu^\ext(\XOR,\epsilon) \le 1-2\epsilon$, since by Theorem \ref{thm:ext-trivial-lb-small-pt-err}, $1 - \IC^\ext_\mu(\XOR,\epsilon) = \IC^\ext_\mu(\XOR,0) - \IC^\ext_\mu(\XOR,\epsilon) \ge 2 \epsilon \IC^\ext_\mu(\XOR,0)  = 2\epsilon$. 

Let $\pi$ be any protocol that computes $[\XOR, \epsilon]$ and $\Pi$ be its random transcript. Let $\ell$ be a leaf of $\pi$ and $\mu_\ell$ denote the distribution conditioned on the protocol reaches $\ell$. By Lemma \ref{lem:IC-over-leaves}, 
\[
\IC^\ext_\mu(\pi) 
= I(XY;\Pi)
= \ent_\mu(XY) - \sum_{\ell} \Pr[\ell] \ent_{\mu_\ell}(XY)
= 1 - \sum_{\ell} \Pr[\ell] \ent_{\mu_\ell}(XY).
\]
We know that $\mu_\ell$ has the same form as $\mu$ (see, e.g., \cite[Section 2.5, Section 5]{my1}), that is,
\begin{equation}  \label{eq:def-mu-ll}
\mu_\ell
=
\begin{pmatrix}
p(\ell) & 0 \\
0 & 1- p(\ell)
\end{pmatrix}
\end{equation}
for some $p(\ell) \in [0,1]$. Hence, $\ent_{\mu_\ell}(XY) = h(p(\ell))$. Thus, our goal is to upper bound $\sum_{\ell} \Pr[\ell] h(p(\ell))$.

As usual, let $p^\pi_{xy}(\ell) = \Pr[\Pi=\ell|XY=xy]$. By Bayes' rule, 
\[
p^\pi_{xy}(\ell)  = \frac{\mu_\ell(xy)}{\mu(xy)} \Pr[\ell] = 2 \mu_\ell(xy) \Pr[\ell]
\]
if $xy=00$ or $xy=11$.  Hence, 
\begin{equation}  \label{eq:apply-rect}
2 \sqrt{\mu_\ell(00) \mu_\ell(11)} \Pr[\ell]
= \sqrt{p^\pi_{00}(\ell)  p^\pi_{11}(\ell)}
= \sqrt{p^\pi_{01}(\ell)  p^\pi_{10}(\ell)}
\le \frac{p^\pi_{01}(\ell) + p^\pi_{10}(\ell)}{2} 
\end{equation}
where we used the  rectangle property Lemma \ref{lem:rectangle-property} in the second equality. Let $L_0$ and $L_1$ denote the set of transcripts with output $0$ and $1$, respectively. Since $\XOR(00) = \XOR(11) = 0$,  $\pi$ computes $[\XOR,\epsilon]$ implies $\Pr[\Pi_{00} \in L_1] \le \epsilon$ and $\Pr[\Pi_{11} \in L_1] \le \epsilon$. Hence, 
\begin{align}  \label{eq:L1}
\sum_{\ell \in L_1} \Pr[\ell] 
&= \sum_{\ell \in L_1} \Big( \mu(00) \Pr[\Pi_{00} = \ell] + \mu(11) \Pr[\Pi_{11} = \ell] \Big) \\
&= \mu(00) \Pr[\Pi_{00} \in L_1] + \mu(11) \Pr[\Pi_{11} \in L_1]
\le \epsilon. \notag 
\end{align}
 Similarly, since  $\XOR(01) = \XOR(10)=1$, one has 
\begin{equation}  \label{eq:aaa}
\sum_{\ell \in L_0} p^\pi_{01}(\ell) = \Pr[\Pi_{01} \in L_0] \le \epsilon, 
\quad
\sum_{\ell \in L_0} p^\pi_{10}(\ell)  = \Pr[\Pi_{10} \in L_0] \le \epsilon. 
\end{equation}
Note that by our notation in \eqref{eq:def-mu-ll}, $p(\ell) = \mu_\ell(00)$ and $1-p(\ell) = \mu_\ell(11)$. Hence, by \eqref{eq:apply-rect} and \eqref{eq:aaa},
\begin{equation}   \label{eq:L0}
\sum_{\ell \in L_0}   2\sqrt{p(\ell) (1-p(\ell))} \Pr[\ell]  
= \sum_{\ell \in L_0}  2\sqrt{\mu_\ell(00) \mu_\ell(11)} \Pr[\ell] 
\le \frac{\sum_{\ell \in L_0} p^\pi_{01}(\ell) + \sum_{\ell \in L_0} p^\pi_{10}(\ell)}{2}
\le \epsilon.
\end{equation}

Therefore, by \eqref{eq:L1}, \eqref{eq:L0}, and Lemma \ref{lem:entropy-small}, we obtain,
\begin{align*}
\sum_{\ell} \Pr[\ell] h(p(\ell)) 
&= \sum_{\ell \in L_0}  \Pr[\ell] h(p(\ell))  + \sum_{\ell \in L_1}  \Pr[\ell] h(p(\ell))   \\
&\le \sum_{\ell \in L_0}  2 \Pr[\ell] \sqrt{p(\ell) (1-p(\ell))} + \sum_{\ell \in L_1}  \Pr[\ell] 
\le 2\epsilon.
\end{align*}
Since $\pi$ is arbitrary, this proves \eqref{eq:xor-goal}.
\end{proof}

It is well-known that $\IC^\ext(\XOR,0) = 2$.  Theorem \ref{thm:XOR-ext-err} and Corollary \ref{cor:pf-ext-pt-err} together imply the following corollary.

\begin{corollary}  \label{cor:XOR-ext-large-err}
For every $0 \le \epsilon \le 1/2$, $1-2 \epsilon \le \IC^\ext(\XOR,\epsilon)  \le 2(1-2\epsilon)$.
Equivalently, $2 \epsilon \le \IC^\ext(\XOR,1/2-\epsilon)  \le 4\epsilon$. 
\end{corollary}

\subsection{The $\AND$ example}  \label{sec:AND}

In \cite{twobounds}, a protocol (see Figure \ref{protocol:AND}) for $\AND$ is proposed that computes $[\AND,1/2-\epsilon]$ and has external information cost $O(\epsilon^2)$ for every distribution. By Lemma \ref{lem:IC-relation}, $\IC(\AND,1/2-\epsilon) \le \IC^\ext(\AND,1/2-\epsilon) = O(\epsilon^2)$, showing that Theorem \ref{thm:large-err-pt-lb} is tight. Since explicit bounds for the information complexity of $\AND$ can be useful (e.g., see \cite{exactComm}), below we obtain such bounds.

\begin{figure}[ht!] 

\begin{framed}
On input $XY$:
\begin{itemize}
\item With probability $2\epsilon$ simply output $0$;
\item With probability $1-2\epsilon$, do the following:
\begin{itemize}
\item Alice sends to Bob $\widetilde{X}$ defined as $\widetilde{X} = X$ with probability $1/2 + 4\epsilon$, and $\widetilde{X} = 1-X$ otherwise;
\item Similarly, Bobs sends to Alice  $\widetilde{Y}$;
\item If both Alice and Bob sends $1$, output $1$;
\item If both Alice and Bob sends $0$, output $0$;
\item Otherwise, output $0$ or $1$ with equal probability $1/2$.
\end{itemize}
\end{itemize}
\end{framed}
\caption{The protocol $\pi$ that computes $[\AND, 1/2-\epsilon]$, from \cite{twobounds}. \label{protocol:AND}}
\end{figure}

Firstly, we obtain a lower bound from Theorem \ref{thm:large-err-pt-lb}.

\begin{corollary}  \label{cor:AND-large-error-lb}
For every distribution $\mu$ such that $\IC_\mu(\AND,0) > 0$, let $\delta>0$ be defined as in Theorem \ref{thm:general-lb-int-cost}. For every $0 \le \epsilon \le 1/2$, 
\[
\IC_\mu(\AND, 1/2-\epsilon) \ge \frac{\delta}{8\ln 2} \epsilon^2.
\]
\end{corollary}

\begin{proof}
By Theorem \ref{thm:char-trivial-distri}, a distribution $\mu$ is internal trivial for $\AND$ if and only if $f_{C_A \times C_B}$ is constant for every connected component $C$ of $\mu$. Since the input space for $\AND$ is $\{0,1\} \times \{0,1\}$, $\supp\mu$ has at most two connected components. It is then easy to verify that internal trivial distributions for $\AND$ are of the following forms:
\begin{equation}  \label{eq:AND-trivial-distri}
\begin{pmatrix}
* & * \\
0 & 0
\end{pmatrix},
\quad
\begin{pmatrix}
* & 0 \\
* & 0
\end{pmatrix},
\quad
\begin{pmatrix}
0 & * \\
* & 0
\end{pmatrix},
\quad
\begin{pmatrix}
* & 0 \\
0 & *
\end{pmatrix}.
\end{equation}
As a result, $\mu$ is not internal trivial for $\AND$ if and only if  $\mu$ is in one of the following forms:
\begin{equation}   \label{eq:AND-non-trivial-distri}
\begin{pmatrix}
0 & 0 \\
+ & +
\end{pmatrix},
\quad
\begin{pmatrix}
0 & + \\
0 & +
\end{pmatrix},
\quad
\begin{pmatrix}
0 & + \\
+ & +
\end{pmatrix},
\quad
\begin{pmatrix}
+ & 0 \\
+ & +
\end{pmatrix},
\quad
\begin{pmatrix}
+ & + \\
0 & +
\end{pmatrix},
\quad
\begin{pmatrix}
+ & + \\
+ & 0
\end{pmatrix},
\quad
\begin{pmatrix}
+ & + \\
+ & +
\end{pmatrix},
\end{equation}
where the ``$+$'' sign indicates the corresponding entry is strictly positive. One can directly verify that these distributions satisfy either Condition (1) or Condition (2) in Theorem \ref{thm:large-err-pt-lb}. 
\end{proof}

For the upper bound, with Wolfram Mathematica we explicitly compute the information cost of the protocol $\pi$ in Figure \ref{protocol:AND}. Let $\mu = 
\begin{pmatrix}
\alpha & \beta \\
\gamma & \delta
\end{pmatrix}
$ be the input distribution. The Wolfram Mathematica computation\footnote{Available upon request.} shows
\begin{equation}  \label{eq:AndResultExt}
\IC_\mu^\ext(\AND, 1/2-\epsilon) \le \IC_\mu^\ext (\pi) = \frac{128(2\alpha\delta + \beta(1-\beta) + \gamma (1-\gamma))}{\ln 2}  \epsilon^2 + O(\epsilon^4).
\end{equation}
and
\begin{equation}   \label{eq:AndResultInt}
\IC_\mu(\AND, 1/2-\epsilon) \le \IC_\mu(\pi)
= \frac{ 128\Big(2\alpha\delta + \beta(1-\beta) + \gamma (1-\gamma)\Big)  \Big((\alpha+\delta)\beta\gamma + (\beta+\gamma)\alpha\delta \Big) }{(\alpha+\beta)(\gamma+\delta)(\alpha+\gamma)(\beta+\delta)\ln 2}\epsilon^2 + O(\epsilon^4).
\end{equation}

Combine this with the lower bound from Corollary \ref{cor:AND-large-error-lb}, we obtain the following. 

\begin{corollary}  \label{cor:AND-large-err-pf}
For every $0 \le \epsilon \le 1/2$ the following hold.
 
For every distribution $\mu$ such that $\IC_\mu(\AND,0) > 0$,
$\IC_\mu(\AND, 1/2-\epsilon) = \Theta(\epsilon^2)$.

For every distribution $\mu$ such that $\IC_\mu^\ext(\AND,0) > 0$, $\IC_\mu^\ext(\AND, 1/2-\epsilon) = \Theta(\epsilon^2)$.

As a result, $\IC(\AND, 1/2-\epsilon) = \Theta(\epsilon^2)$ and 
$\IC^\ext(\AND, 1/2-\epsilon) = \Theta(\epsilon^2)$.
\end{corollary}

This contrasts to $\XOR$ since Corollary \ref{cor:XOR-ext-large-err} shows $\IC^\ext(\XOR, 1/2-\epsilon) = \Theta(\epsilon)$ for every $0\le \epsilon \le 1/2$. 

\section{Discussion and open problems}   \label{sec:prob}
Since $h(\epsilon) - h(0) = h(\epsilon)$ and $h(1/2) - h(1/2-\epsilon) = \Theta(\epsilon^2)$ when $\epsilon>0$ is sufficiently small, Theorem \ref{thm:ext-prod-lb-small-pt-err} and Theorem \ref{thm:large-err-pt-lb} show that the behaviour of information complexity when error is allowed has some similarities to the behaviour of Shannon entropy function. However, the tight examples in Section \ref{sec:eg} show that the behaviour of information complexity is more complicated. Solving the following problems will shed more light on trading information complexity for error. 

{\noindent \bf Problem 1}: Is  the upper bound $\IC^\ext_\mu(f,0) -\IC_\mu^\ext(f,\epsilon) = O(h(\sqrt{\epsilon}))$ tight (in terms of the order of $\epsilon$)?

{\noindent \bf Problem 2}: In viewing of Theorem \ref{thm:ext-prod-lb-small-pt-err}, maybe the lower bound for prior-free external information complexity could be improved. For example, is $\IC^\ext(f,0) - \IC^\ext(f,\epsilon) \ge \Omega(h(\epsilon))$ true?

{\noindent \bf Problem 3}: Remove the conditions in Theorem \ref{thm:large-err-pt-lb}. One example that is not included in Theorem \ref{thm:large-err-pt-lb} is
\begin{equation}
\begin{pmatrix}
0' & 0' & 0 \\
0 & 0' & 0' \\
1 & 0 & 0' \\
\end{pmatrix},
\end{equation}
where $\supp\mu$ consists of $0'$s. Perhaps the same lower bound holds.

{\noindent \bf Problem 4}: We showed in Corollary \ref{cor:XOR-ext-large-err} that $2 \epsilon \le \IC^\ext(\XOR,1/2-\epsilon)  \le 4\epsilon$.  It seems an elegant problem to determine $\IC^\ext(\XOR,\epsilon)$ \emph{exactly} for every $0\le \epsilon \le 1/2$.

Another natural direction is to generalize \cite{my1} and the present work to multi-party information complexity, and to quantum information complexity defined in \cite{qic}. Lastly, considering the wide applicability of information complexity (such as to communication complexity, data stream, decision tree complexity, extension complexity, etc, as mentioned in the introduction), it would be great to see new applications based on the techniques and results developed in \cite{my1} and the present work.

\bibliography{mybib}{}
\bibliographystyle{amsalpha}

\end{document}